\documentclass[runningheads,twoside]{llncs}

\usepackage[utf8]{inputenc}
\usepackage[T1]{fontenc}

\usepackage{graphicx}
\usepackage{amssymb}
\usepackage{textcomp}
\usepackage{amscd}
\usepackage{amsmath}
\usepackage[usenames]{color}
\definecolor{sonnengelb}{rgb}{0.9,0.5,0} 
\definecolor{brown}{rgb}{0.5,1,0}
\usepackage{todonotes}
\usepackage{enumitem} 
\usepackage{caption}
\usepackage{subcaption}

\usepackage{algorithm,algorithmic}

\begin{document}
\mainmatter              
\title{Broder's Chain Is Not Rapidly Mixing\thanks{This work
    was supported by the DFG Focus Program Algorithm Engineering,
    grant MU 1482/4-3.}}
\titlerunning{Sampling of Perfect Matchings}
\author{Annabell Berger and Steffen Rechner}
\institute{Department of Computer Science\\
  Martin-Luther-Universit\"at Halle-Wittenberg\\
\email{\{berger, rechner\}@informatik.uni-halle.de}}
\authorrunning{A. Berger and S. Rechner}
\maketitle

\begin{abstract}
We prove that Broder's Markov chain for approximate sampling near-perfect and perfect matchings is not rapidly mixing for Hamiltonian, regular, threshold and planar bipartite graphs, filling a gap in the literature. In the second part we experimentally compare Broder's chain with the Markov chain by Jerrum, Sinclair and Vigoda from 2004. For the first time, we provide a systematic experimental investigation of mixing time bounds for these Markov chains. We observe that the exact total mixing time is in many cases significantly lower than known upper bounds using canonical path or multicommodity flow methods, even if the structure of an underlying state graph is known. In contrast we observe comparatively tighter upper bounds using spectral gaps. \\

\end{abstract}
 
\section{Introduction}\label{Section:Introduction}

\paragraph{Uniformly Generating Perfect Matchings}
The uniform generation of matchings or perfect matchings in bipartite graphs are classical and well-studied problems in combinatorial optimization. Sampling matchings is also an important tool for statistical physics (there called \emph{monomer-dimer systems}, see \cite{Heilmann04}). The problem demonstrates advantages but also symptomatic difficulties of current sampling techniques. Let us briefly summarize open questions and typical difficulties. For a more general overview consider the publications of Jerrum, Sinclair and Vigoda \cite{Jerrum:1989:AP:76071.76077,JerrumSinclairVigoda04}. The problems of \emph{counting matchings} and \emph{counting perfect matchings} belong to the class of~$\#P$-complete problems \cite{Valiant79}. On the contrary they are \emph{self-reducible}, i.e., the solution set~$S$ can be expressed in terms of a polynomially bounded number of solution sets~$S_i,$ such that each~$S_i$ belongs to a smaller problem instance. Self-reducible problems are interesting with respect to a result of Jerrum, Vazirani and Valiant \cite{JerrumVV86}, proving the equivalence of the existence of a \emph{fully-polynomial randomized approximation scheme} (FPRAS) and a \emph{fully polynomial almost uniform sampler} (FPAUS). Jerrum and Sinclair \cite{Jerrum:1989:AP:76071.76077} proved in 1989 the existence of an FPAUS for matchings in bipartite graphs. Jerrum, Sinclair and Vigoda \cite{JerrumSinclairVigoda04} constructed in 2004 an FPAUS for perfect matchings. Hence, the associated counting problems have an FPRAS, because they are self-reducible. Note that exact counting and uniform sampling perfect matchings in planar graphs is efficiently possible \cite{Kasteleyn61,Edmonds1965a}. In contrast exact counting all matchings is~$\#P$-complete in planar graphs \cite{Jerrum87}.
\paragraph{Metropolis Markov chains for Sampling Matchings}\label{paragraph:MetropolisChains}
One important tool for approximate sampling are \emph{Metropolis Markov chains} which can be seen as a random walk on a set $\Omega$ of combinatorial objects (for example matchings), the so-called \emph{states}, which are connected by a given \emph{neighborhood-structure}, i.e., two states $x,y \in \Omega$ are \emph{neighbored} if they differ by a local change. This definition induces a so-called \emph{state graph} $\Gamma=(\Omega,\Psi)$ representing the objects and their adjacencies. Furthermore, we define for each neighbored vertex pair~$x,y$ in~$\Gamma$ the so-called \emph{proposal probability} $\kappa(x,y)$ and a \emph{weight function} $w(x)$ for each $x \in \Omega$. A step from~$x$ to~$y$ is done with \emph{transition probability} 
\begin{eqnarray}
P(x,y)=\kappa(x,y)\min{ \left(1,\frac{w(y)}{w(x)} \right)}.
\label{eqn:transition_probability}
\end{eqnarray}
The matrix $P:=(P(x,y))_{x,y \in \Omega}$ is called \emph{transition matrix}. \emph{The fundamental theorem for Markov chains} says that the chain in Algorithm~\ref{alg:Metropolis} converges for $t \to \infty$ to the unique, \emph{stationary distribution}~$\pi$, if state graph $\Gamma$ is non-bipartite, connected and \emph{reversible} with respect to~$\pi,$
i.e., $\pi(x)P(x,y)=\pi(y)P(y,x).$

\begin{algorithm}[H]
    \caption{Metropolis Markov chain} \label{alg:Metropolis}
    \begin{algorithmic}[1]
      \REQUIRE~$x \in \Omega$,~$\Psi \subset \Omega \times \Omega$,~$t$,~$\kappa$,~$w$ \qquad  //$t$ denotes the number of steps.
      \ENSURE~$y \in \Omega$ with probability~$\pi(y)=\frac{w(y)}{\sum_{x \in \Omega}w(x)}$.
      \FOR{i=1 to t}
         \STATE (i) Pick a neighbor~$y$ of~$x$ with probability~$\kappa(x,y).$
        	\STATE (ii)~$x \leftarrow y$ with probability~$\min{ \left( 1, \frac{w(y)}{w(x)} \right)}.$
      \ENDFOR
      \end{algorithmic}
      \label{alg:mmc}
  \end{algorithm}

We investigate three Metropolis Markov chains, the \emph{monomer-dimer-chain} \cite{Heilmann04,Jerrum:1989:AP:76071.76077}, \emph{Broder's chain} \cite{broder86,Jerrum:1989:AP:76071.76077} and the \emph{JSV-chain} \cite{JerrumSinclairVigoda04}. The monomer-dimer chain is used for uniformly sampling matchings $\mathfrak{M}(G)$ in a bipartite graph $G=(U \cup V,E)$ with $2n$ vertices. The other two chains are for uniformly sampling near-perfect matchings $N(G)$ and perfect matchings $M(G)$ in~$G$. Let $N_{uv}(G) \subset N(G)$ be the subset of near-perfect matchings such that the  vertices~$u,v$ are unmatched. For simplicity, we often set $N_{uv}:=N_{u,v}(G)$, when graph~$G$ is unambiguous. In all three chains two matchings~$M,M'$ are neighbored if they differ by (a) exactly one edge, (b) by two adjacent edges, i.e., the symmetric difference is~$M \triangle M'=\{e=(u,v),e'=(u',v)\}$ or (c)~$M=M'.$ When choosing an edge~$e'=(u',v)$ one has to decide if adding or deleting of~$e'$ in~$M$ leads to a neighbor~$M'$. In the JSV-chain the idea is mainly the same, the set of possible~$e'$ differs from the first two steps. For the exact definition of the chains consider \cite{Jerrum:1989:AP:76071.76077,JerrumSinclairVigoda04}. Much more interesting is the setting of $w(x)=1$ for all $x \in \Omega$ for the monomer-dimer chain and for Broder's chain which simplifies equation~(\ref{eqn:transition_probability}), but leads to a major disadvantage. The problems arise if the fraction~$|N(G)|/|M(G)|$ is too large, e.g. $|N(G)|/|M(G)| > 2^n$, because the chain samples all near-perfect and perfect matchings with the same probability so the expected number of trials to find a perfect matching is in $\mathcal{O}(2^n)$. This problem has been overcome by Jerrum, Sinclair and Vigoda in  the JSV-chain defining 
\begin{equation} w(M):=\begin{cases}1,&M \in M(G),\\ \frac{|M(G)|}{|N_{uv}(G)|},&M \in N_{uv}(G), \end{cases} \label{eqn:weight}\end{equation}
leading to $w(N(G))/w(M(G))=n^2$. Note that all perfect matchings are uniformly distributed, but~$\pi$ is not the uniform distribution. Unfortunately, computing the exact value of~$w$ is~$\#P$-complete as we proved (see Proposition~\ref{prop1} in the appendix). Jerrum et al. \cite{JerrumSinclairVigoda04}  propose a simulated annealing approach with asymptotic running time $O(n^{11}\log^2(n) \log(1/\eta))$ for a correctness probability $1-\eta$. Bez{\'a}kov{\'a} et al. \cite{Bezakova06acceleratingsimulated} improved it to $O(n^7 \log^4(n)).$ This time dominates the running time for the JSV-chain and is for many applications orders of magnitude too large.

\paragraph{Mixing Time and Upper Bounds}
The main question about Markov chains is their efficiency. How fast does a Metropolis chain converge to its stationary distribution~$\pi$? We denote the probability distribution of a chain at time~$t$ with initial state~$x \in \Omega$ by~$p^t_x.$ The \emph{variation distance} is given by~$d_{tv}(\pi,p_x^t):=\frac{1}{2}\sum_{y \in \Omega}|\pi(y)-p_x^t(y)|$. The \emph{total variation distance}~$d(\pi,t):= \max_{x \in \Omega}d_{tv}(\pi,p^t_x) \label{eqn:total_variation_distance}$ is a certain kind of ``worst case'' variation distance not depending on the initial state in~$\Gamma$. The \emph{total mixing time} is defined by $\tau(\epsilon) = \min \{~t~|~d(\pi,t) \leq \epsilon\}$.
Let~$1 = \lambda_1 > \lambda_2 \geq \ldots \geq \lambda_{|\Omega|} > -1$ be the eigenvalues of transition matrix~$P$ and $\lambda_{\max}:= \max \{ |\lambda_2|, |\lambda_{\Omega}| \}$. The total mixing time can be bounded by the \emph{spectral bound} \cite{Sinclair92}, i.e.,
\begin{eqnarray}
\tau(\epsilon) \leq \left(1-\lambda_{\max}\right)^{-1} \cdot \left( \ln(\epsilon^{-1}) + \ln(\pi_{\min}^{-1}) \right).
\label{eqn:spectral_bound}
\end{eqnarray}
Note that~$\pi_{min}$ denotes the smallest component of~$\pi.$ Sinclair's \emph{multicommodity flow method} is often used for bounding the mixing time. Let~$\mathcal{P} = \bigcup_{x\not=y} \mathcal{P}_{xy}$ be a family of simple paths in~$\Gamma$, each~$\mathcal{P}_{xy}$ consisting of simple paths between~$x$ and~$y \in \Omega$. Following \cite{Sinclair92}, a \emph{flow} is a function~$f \colon \mathcal{P} \to \mathbb{R_+}$. We define two flow functions~$f_1$ and~$f_2$ satisfying for all~$x,y \in \Omega, x \not= y$,
\[
\sum_{p \in \mathcal{P}_{xy}} f_1(p) = \pi(x)\pi(y) \quad  \text{ resp. } \quad \sum_{p \in \mathcal{P}_{xy}} f_2(p) = \pi(x)\pi(y)|p|.\] 
The \emph{maximum loading} $\rho_i$ ($i \in \{1,2\}$) for an arc~$a$ in~$\Gamma$ with respect to~$\mathcal{P}$ is then defined as 
\begin{equation}
\rho_i(\mathcal{P}):=\max_{a \in \Psi} (f_i(a)/Q(a))\label{eqn:maximumLoading}
\end{equation}
with~$Q(a) := Q(u,v) = \pi(u) P(u,v)$ for an arc~$a=(u,v) \in \Psi$ in~$\Gamma$ and $f_i(a): = \sum_{p \in \mathcal{P} \colon a \in p} f_i(p).$ 
By \cite[Theorem~3', Proposition~1~(i)]{Sinclair92} and \cite[Theorem~5', Proposition~1~(i)]{Sinclair92} for any system of paths~$\mathcal{P}$  the mixing time of a reversible Markov chain where $\lambda_{\max} = |\lambda_2|$ can be bounded by the \emph{multicommodity bound}, i.e.,
\begin{eqnarray}
\tau(\epsilon) &\leq& \rho_2(\mathcal{P}) \cdot \left( \ln(\epsilon^{-1}) + \ln(\pi_{\min}^{-1}) \right)\label{eqn:flow2}.
\end{eqnarray}
The quality of the multicommodity bounds depend on $\mathcal{P}$ and will be discussed in Section~\ref{sec:experiments}. The multicommodity bound (\ref{eqn:flow2}) using the maximum loading $\rho_2$ is always better than an analogous result with $\rho_1$. We need the definition of $\rho_1$ for lower bounding the mixing times in our proofs. If $\tau(\epsilon) \leq p(n)( \ln(\epsilon^{-1}) + \ln(\pi_{\min}^{-1}))$ for each input instance of size $n$, where $p$ is a polynomial which depends on~$n$, a Metropolis Markov chain is called \emph{rapidly mixing}. In Table~\ref{table:mixingtimes} an overview of the best known upper bounds for mixing times. We observe that the mixing time of Broder's chain is way too large to be practicable and depends on the fraction $|N(G)|/|M(G)|$. The mixing time of the JSV-chain is not that large, but the computation of weights~$w$ becomes the bottleneck and limits its practical applicability as explained above. 

\begin{table}[t]
	\centering
    \begin{tabular}{ | l | l | l | l |}
    \hline
    chain & graph class &rapidly mixing & best known bound of $\tau(\epsilon)$\\ \hline
    Broder's & bipartite & no (this paper) & $16^2|E|^2\left(\frac{|N(G)|}{|M(G)|}\right)^4\ln(|\Omega|\cdot\epsilon^{-1})$ \cite{Sinclair:1993:ARG:140552} \\
     & dense bipartite & yes \cite{Jerrum:1989:AP:76071.76077} & $6 \cdot n^7 \ln(|\Omega|\cdot\epsilon^{-1})$ \cite{Diaconis99statisticalproblems} \\
    JSV& bipartite &yes \cite{JerrumSinclairVigoda04}& $n^4 \ln((\pi_{min}\cdot\epsilon)^{-1})$ \cite{Bezakova06acceleratingsimulated}\\ \hline
    \end{tabular}
    \caption{Upper bounds of $\tau(\epsilon)$.}
    \label{table:mixingtimes}
\end{table}

\paragraph{Motivation and Contribution}\label{Paragraph:MotivationAndContribution}
The main cause for the impracticality of rapidly mixing Markov chains are the high degree polynomial bounds of the mixing time.
Figure~\ref{figure:MixingboundSinclair} shows the exact total mixing time of Broder's chain for a small example of a bipartite graph with six vertices. Note that this graph is dense in the sense that its minimal vertex degree is not less than $n/2$, so we can apply $6n^7\cdot\ln\left( |\Omega| \epsilon^{-1} \right)$ from Table~\ref{table:mixingtimes} as an upper bound. Clearly, the upper bound differs from the exact total mixing time by several orders of magnitude. We identify two possible explanations for the large difference.

\begin{figure}[t]
	\centering
	\begin{subfigure}[b]{0.25\textwidth}
		\centering
		\includegraphics[width=0.3\textwidth]{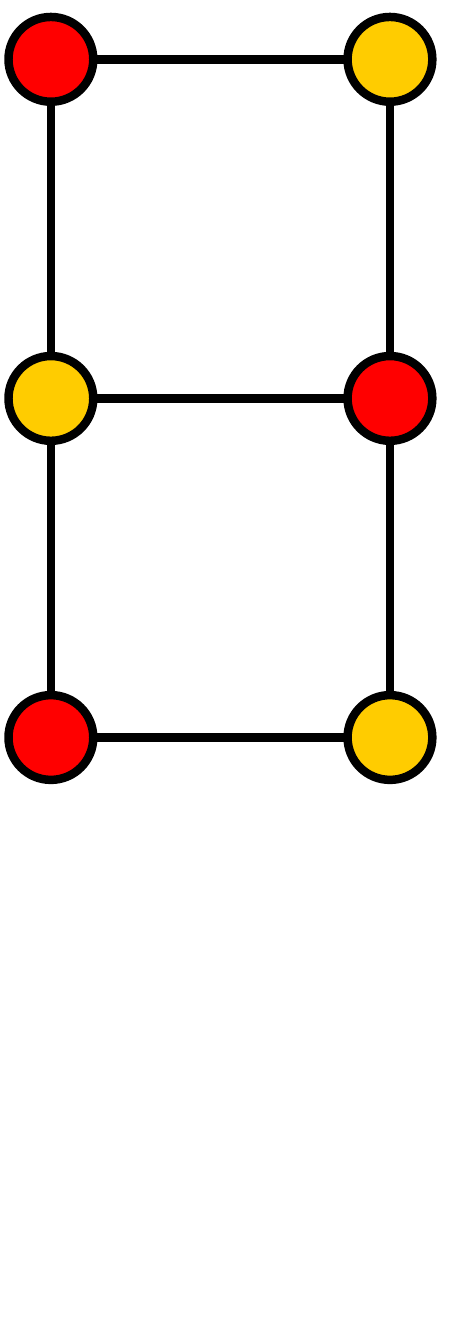}
		\caption{Bipartite graph~$G$}
		\label{figure:bipExample}
	\end{subfigure}
	~
	\begin{subfigure}[b]{0.7\textwidth}
		\centering
		\includegraphics[width=0.65\textwidth]{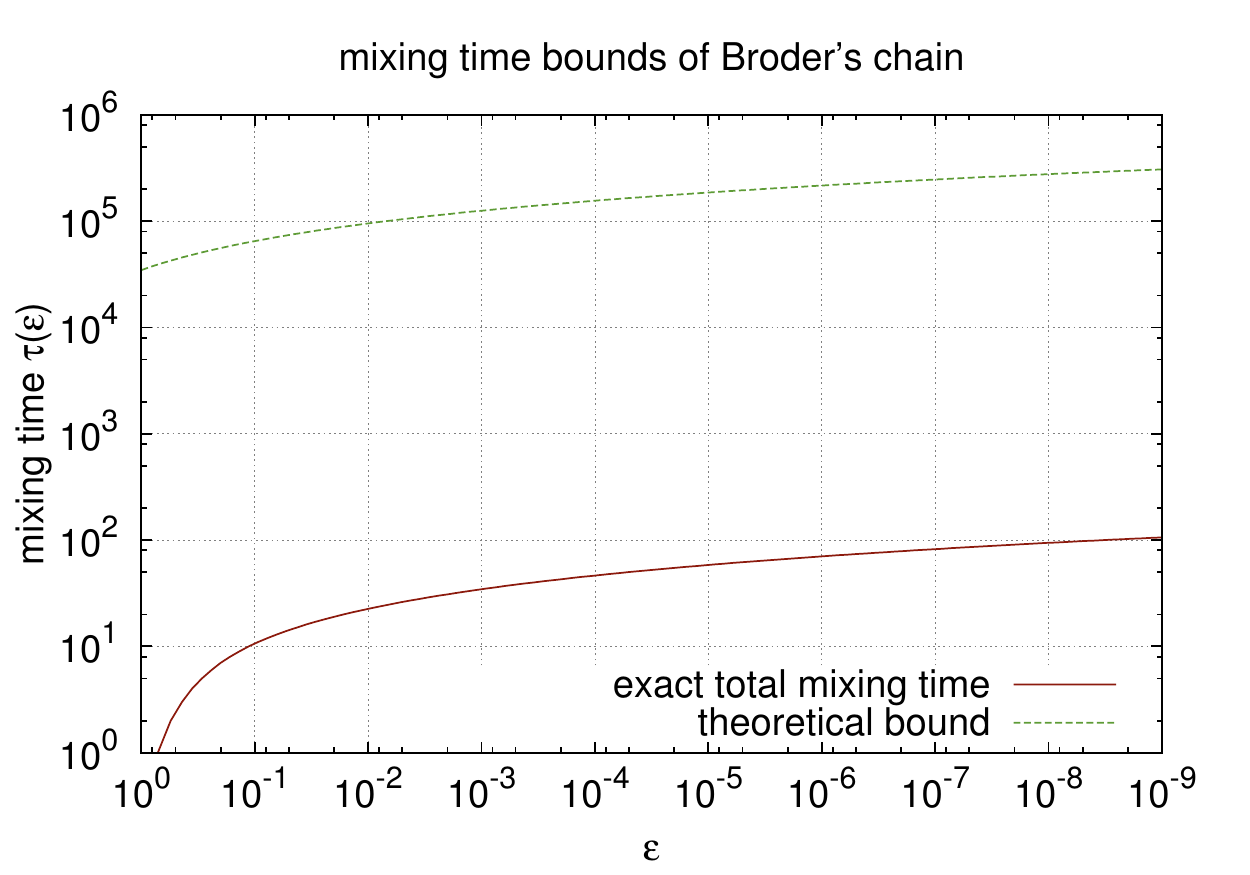}
		\caption{Mixing time versus theoretical bound}
	\end{subfigure}
	\caption{Example graph and corresponding mixing time of Broder's chain}
	\label{figure:MixingboundSinclair}
\end{figure}

\begin{figure}[t]
\centering
\includegraphics[width=\textwidth]{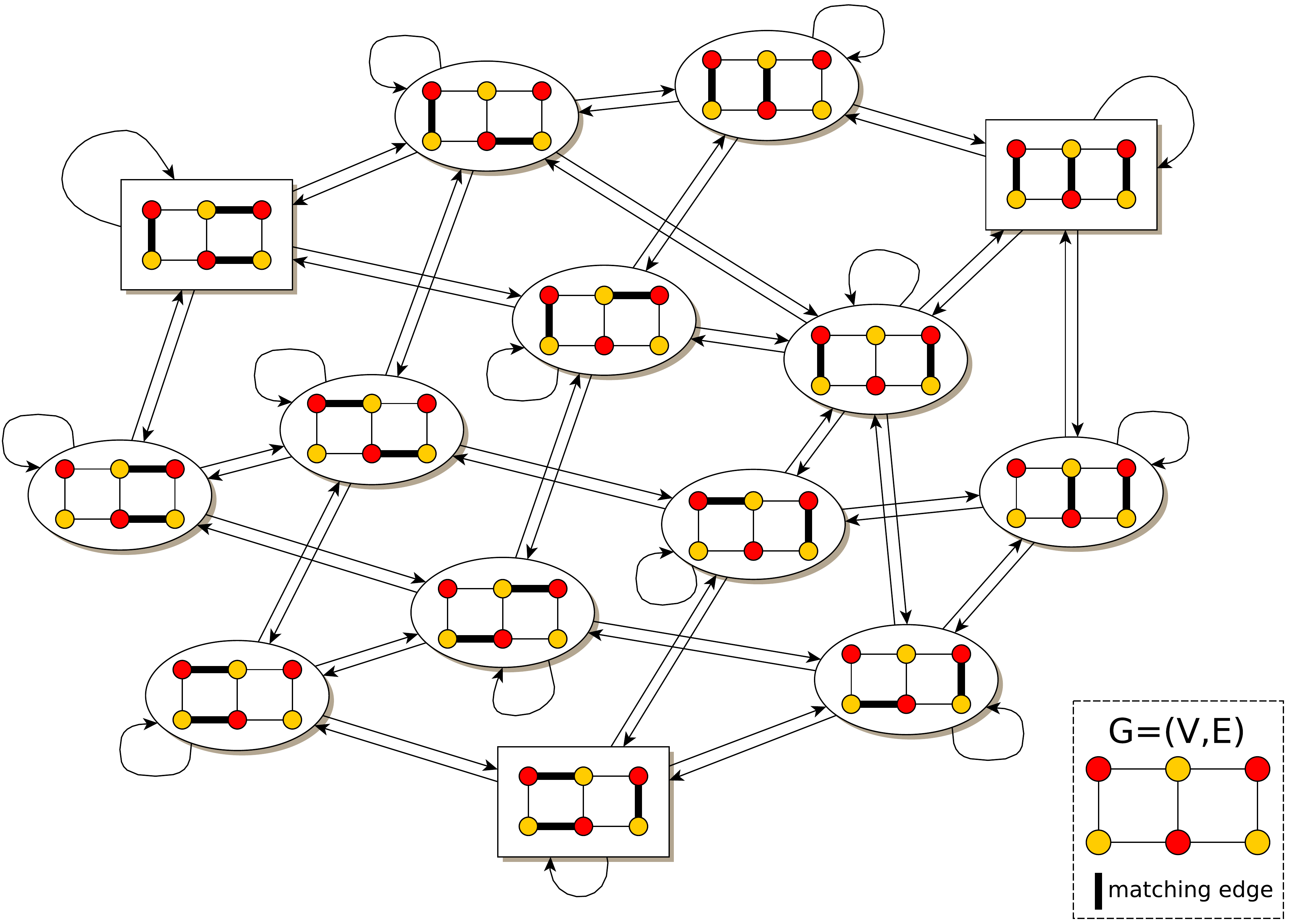}
\caption{Corresponding state graph to bipartite graph in Figure~\ref{figure:bipExample}. Three perfect matchings (rectangular), and 11 near-perfect matchings (elliptic) connected by transition rules of Broder's chain. Transition probabilities omitted}
\label{figure:stategraph}
\end{figure}

\begin{enumerate}
\item Does the difference result from a lack of structural insights about state graphs while applying the bounding methods? Would applying the methods to a known state graph lead to tighter bounds?
\item Does the difference result from a possible weakness of the bounding methods and still occurs when the structure of~$\Gamma$ is given explicitly?
\end{enumerate}

We try to give experimental answers to these questions. For all bipartite graphs with up to~$12$ vertices we computed the corresponding state graphs and determined exact total mixing times and several upper bounds as defined in the previous paragraph.
Note that in practice a state graph is never constructed explicitly because one only needs to know the current state of~$\Omega$ and how to construct a neighbor. In theoretical considerations, one also does not know the explicit structure of a state graph and would possibly try to prove common properties for a better application of bounding methods. In contrast, in our experiments the state graphs are constructed explicitly. We observe that the multicommodity bound depends on the set of paths~$\mathcal{P}$ and is significantly larger than the total mixing time, whereas the spectral bound seems to be much tighter.  We recommend for future work to investigate the structure of state graphs combined with additional research in spectral graph theory.\\
Up to this paper it has been open, whether Broder's chain is rapidly mixing or not. Sinclair proved in \cite{Sinclair:1993:ARG:140552} (see Table~\ref{table:mixingtimes}) that the mixing time of a graph~$G$ corresponds to the fraction~$|N(G)|/|M(G)|$. For dense graphs Broder showed that~$|N(G)|/|M(G)| \leq n^2.$ We prove $\#P$-complete\-ness for the computation of this value (see Theorem~\ref{SharpComplete} in the appendix). Moreover, for the first time we prove that Broder's chain is not rapidly mixing in general. Especially, we analyze several classes of graphs (regular, planar, threshold, Hamiltonian) for this property. Interestingly, planar graphs \cite{Edmonds1965a} and threshold graphs can be sampled efficiently using other methods. For threshold graphs we propose a simple and efficient approach based on ideas of Brualdi and Ryser \cite[Corollary 7.2.6.]{brualdi:91}. It is a remarkable weakness of Broder's chain that it cannot rapidly sample a (near)-perfect matching in planar or threshold graphs. Future research should be concentrated on the construction of more efficient techniques to estimate the weights $w$ for the JSV-chain, especially for special graph classes.

\paragraph{Overview}   
 In Section \ref{sec:Broder'sChain} we prove that Broder's chain is not rapidly mixing. In Section~\ref{sec:experiments} we make experiments on the total mixing times and several bounds of Broder's chain and the JSV-chain. In the appendix we prove the~$\#P$-completeness for the computation of~$\frac{|N(G)|}{|M(G)|}.$

\section{Large Mixing Times for Broder's Chain}
\label{sec:Broder'sChain}

In this section we show that Broder's chain is not rapidly mixing. For a long time it has been known that this chain cannot be used for efficiently sampling \emph{perfect} matchings in graph~$G$ with $2n$ vertices if the fraction~$|N(G)|/|M(G)|$ is too large (e.g. $|N(G)|/|M(G)| > 2^n$). The reason is simply that one needs too many trials~$l$ to get a perfect matching whether or not the chain is rapidly mixing. Note that formally Broder's chain does not solve the problem of sampling perfect matchings but of sampling \emph{near-perfect and perfect} matchings. Hence, asking whether this chain has efficient mixing time is different from the problem of efficiently sampling a perfect matching with Broder's chain. A positive answer could potentially improve the total running time for sampling a perfect matching for some graph class $\mathfrak{G}$ where~$|N(G)|/|M(G)|\leq n^{k}$ for $G \in \mathfrak{G}$ is bounded polynomially for some constant $k$. We explain the idea for sampling perfect matchings assuming Broder's chain were rapidly mixing.

\begin{enumerate}
\item Apply $l=n^k$ times the monomer-dimer-chain using Algorithm~\ref{alg:Metropolis} with~$t:=\tau(\epsilon)= \lceil~8n \cdot |E| (\ln(|\Omega|)+\ln(\epsilon^{-1}))\rceil$ (best bound on mixing time in Table~\ref{table:mixingtimes}) and a given~$\epsilon$.
\item Stop sampling for trial $l' \leq l$ if the $t$-th matching is a perfect matching.
\end{enumerate}
Clearly, if this algorithm creates in one of the~$l$ trials in the~$t$-th step a perfect matching we are done. This approach has constant success probability when fraction~$|\mathfrak{M}(G)|/|M(G)| \leq n^{k}$. We cannot compute this fraction, but the efficient mixing time guarantees the correctness of an uniform perfect matching if we find one. This approach makes only practical sense if the running time is better than the worst case running time of the JSV-chain (mixing time plus running time of simulated annealing see the second paragraph in Section~\ref{Section:Introduction}), but in this case we find a more efficient method. Since~$\tau(\epsilon) \in O(nm\ln(|\Omega|\cdot\epsilon^{-1}))$ (see Table~\ref{table:mixingtimes}) and $|\Omega| \leq 2^{n^2}$ we get $\tau(\epsilon) \in O(n^3m)$ and would set~$l:= \frac{n^4}{m} \log^4(n).$ The disadvantage here is the fraction $|\mathfrak{M}(G)|/|M(G)|\leq n^{k}$ which is more improbably than fraction~$\frac{|N(G)|}{|M(G)|}\leq n^{k}$ for Broder's chain. Sinclair \cite{Sinclair:1993:ARG:140552} proved for such cases that Broder's chain is rapidly mixing (see Table~\ref{table:mixingtimes}). However, computing this ratio is~$\#P$-complete (see Theorem~\ref{SharpComplete} in the appendix) for a given bipartite graph. Hence, for a given graph we cannot use the result of Sinclair. This approach makes only practical sense for Broder's chain if the running time $n^{k'+k}(\ln(|\Omega|)+\ln(\epsilon^{-1}))$ is better than the worst case running time $O(n^7 \log^4(n))$ of the JSV-chain (mixing time plus running time of simulated annealing, see Section~\ref{Section:Introduction}). Note, for monomer-dimer-chains it can happen that the fraction~$\mathfrak{M}(G)/M(G)$ cannot be bounded polynomially but the chain will still be rapidly mixing. To the best of our knowledge there does not exist a general result whether Broder's chain is rapidly mixing or not. 

\begin{theorem}
Broder's chain is not rapidly mixing for planar, regular, Hamiltonian and threshold graphs.
\label{theorem:broder_not_rapidly}
\end{theorem}

We start with the planar \emph{hexagon graph} class (see Figure~\ref{figure:hexagon}) introduced by Jerrum et al. \cite{JerrumSinclairVigoda04}. Such a graph consists of~$k$ hexagons~$H_1,\ldots,H_k$ connected in a chain by adding a single edge between adjacent hexagons~$H_i$ and~$H_{i+1}$ for~$i \in \{ 1, \ldots, k-1\}$. The left most hexagon $H_1$ and right most hexagon $H_k$ are connected to vertices~$u$ and~$v$. The number of vertices of a hexagon graph is~$2n=6k+2$. A bipartite graph of this class possesses exactly one perfect matching. Moreover, we find that $|N_{uv}|=2^k$. To prove that the chain is not rapidly mixing we need the following result on lower bounds which easily follows from the results in \cite[Proposition~1 (ii), Corollary~9]{Sinclair92} by some elementary transformations. We apply Corollary~9 in this paper on Proposition~1 which is given there for the second largest eigenvalue. (Clearly, the lower bound of the second largest eigenvalue~$\lambda_1$ is also a lower bound for~$\lambda_{max}$.) 

\begin{figure}[tbp]
\centering
\includegraphics[width=8cm]{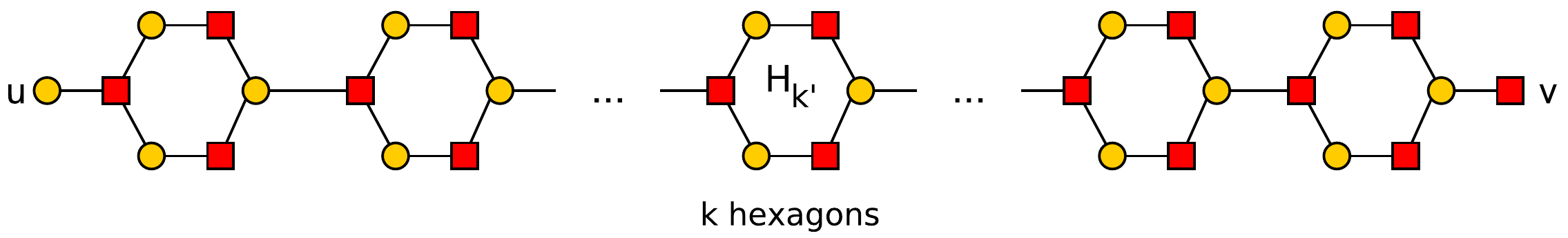}
\caption{Hexagon graph class (planar)}
\label{figure:hexagon}
\end{figure}

\begin{proposition}\label{Proposition:lowerBoundMixingtime}
For any reversible Markov chain with stationary distribution~$\pi$ we have \[\tau(\epsilon) \in \Omega\left(\frac{\rho_1(\mathcal{P}) \ln(2\epsilon^{-1})}{\ln{\pi_{\min}^{-1}}}\right).\]
\end{proposition}

Set~$\mathcal{P}$ here denotes an arbitrary set of simple paths in state graph~$\Gamma$ with exactly one path~$\mathcal{P}_{MM'}$ between each vertex pair~$M$ and~$M'$. Assume $a=(\bar{M},\bar{M}') \in \Psi$ with~$\bar{M} \neq \bar{M}'$ is an arc for which $\rho_1(\mathcal{P})=f_1(a)/Q(a)$ in equation (\ref{eqn:maximumLoading}). Then each path~$\mathcal{P}_{MM'} \in \mathcal{P}$ using arc $a$ in Broder's chain contributes to $\rho_1$ a quantity of 
\begin{equation}
\frac{f_1(\mathcal{P}_{MM'})}{Q((\bar{M},\bar{M}'))} = \frac{\pi(M)\pi(M')}{\pi(\bar{M})P(\bar{M},\bar{M}')} = \frac{m}{|\Omega|}.
\label{eqn:ContributionMaximumLoading}
\end{equation}
The idea of our proofs works as follows. We show that for each choice of a set~$\mathcal{P}$ the maximum loading~$\rho_1(\mathcal{P})$ cannot be bounded from below by a polynomial. Note that a near-perfect matching in~$N_{uv}$ consists of~$k$ perfectly matched hexagons, where each hexagon possesses two possibilities for a perfect matching which we denote by~$M_1$ and~$M_2$. We now partition the set~$N_{uv}$ into the subsets~$N_{uv}^1$ and~$N_{uv}^2$, such that the ``middle'' hexagon subgraph~$H_{k'}$ with~$k' = \left\lceil k/2 \right\rceil$ contains in each near-perfect matching of~$N_{uv}^1$ a perfect matching of type~$M_1$ and in each near-perfect matching in~$N_{uv}^2$ a perfect matching of type~$M_2$. Clearly,~$|N_{uv}^1|=|N_{uv}^2|=2^{k-1}.$ Now, we consider a possible simple path $\mathcal{P}_{MM'}$ in the corresponding state graph~$\Gamma$ of the hexagon graph between an arbitrary vertex pair~$(M,M')\in N_{uv}^1\times N_{uv}^2.$ Since~$H_{k'}$ contains~$M_1$ in~$M$ and~$M_2$ in~$M'$, there has to be on $\mathcal{P}_{MM'}$ at least one matching~$M \in N_{u^*v'} \cup N_{u'v^*}$ ($u' \in U_{k'}, v' \in V_{k'}, u^* \in U, v^* \in V$), respectively. Otherwise, the transition rules of Broder's chain could not convert the matching~$M_1$ into~$M_2.$ On the other hand it is not difficult to see that~$|N_{u^*v'}|\leq 2^{k-k'}\leq 2^{\frac{k}{2}}$ for all~$u^* \in V$ and~$v' \in V_{k'}.$ For symmetry reasons the same is true for~$N_{u'v^*}$.
If we now consider an arbitrary set~$\mathcal{P}'$ of simple paths between all vertex pairs~$(M,M')$ of~$N_{uv}^1\times N_{uv}^2$ (exactly one between each pair), then we find that~$|\mathcal{P}'|=2^{2k-2}$ paths use an arc in~$\Gamma$ which is adjacent to the vertex set~$S:=(\bigcup_{u^* \in U, v'\in V_{k'}} N_{u^*,v'}) \cup (\bigcup_{v^* \in V, u'\in U_{k'}} N_{u'v^*}).$
We find that~$|S|\leq 6n 2^{\frac{k}{2}}.$ Furthermore, each vertex in~$S$ is adjacent to~$m$ arcs in~$\Gamma$ (some of them can be loops). Hence, there exists an arc~$a$ in~$\Gamma$ which has to be used from at least~$(2^{2k-2})/(6nm2^{\frac{k}{2}}) = (2^{\frac{3}{2}k-2}) / (6mn)$ paths of~$\mathcal{P}'$. Note that~$|N_{uv}|\geq |N_{u^*v^*}|$ for all~$u^* \in U$ and~$v^* \in V$ and there exist~$n^2$ sets~$N_{u^*v^*}.$  With (\ref{eqn:ContributionMaximumLoading}) we get for an arbitrary $\mathcal{P}$ a maximum loading of $\rho_1(\mathcal{P})\geq \rho_1(\mathcal{P'})\geq$
\footnotesize
\[\frac{2^{\frac{3}{2}k-2}}{6mn} \cdot \frac{m}{|\Omega|} = \frac{2^{\frac{3}{2}k-2} }{6n\cdot \left( \sum_{u^* \in U, v^* \in V} |N_{u^*, v^*}| + |M| \right)} \geq \frac{2^{\frac{3}{2}k-2}}{6n\left( n^2 + 1 \right) 2^k} = \frac{2^{\frac{n}{6}-\frac{13}{6}}}{6n\left( n^2 + 1 \right).}
\]
\normalsize
With Proposition~\ref{Proposition:lowerBoundMixingtime} we find the following result.

\begin{proposition}\label{Proposition:hexagonGraph}
Broder's chain is not rapidly mixing for hexagon graphs.
\end{proposition}

In the following we give three additional graph classes not possessing a polynomial mixing time. The first is the class of bipartite threshold graphs, i.e. graphs such that there does not exist a further graph possessing the same vertex degrees. For further characterizations consider the book of Mahadev and Peled \cite{MaPe:95}. The bi-adjacency matrix of such a graph has the special property that each row starts with a number of ``one'' entries corresponding to its vertex degree. The rest of each row is filled by zeroes. (A matrix of this form is called \emph{Ferrers matrix} introduced by Ferrers in the $19$th century, see \cite{Syl:1882}.) Consider the matrix $B$ in Figure~\ref{figure:thresholdGraph} as example for the bi-adjacency matrix of a threshold graph~$G=(U \cup V,E)$ where the vertices of $U$ correspond to the rows and the vertices of $V$ to the columns.

\begin{figure}
\begin{subfigure}[b]{0.47\textwidth}
\centering
\includegraphics[width=0.4\textwidth]{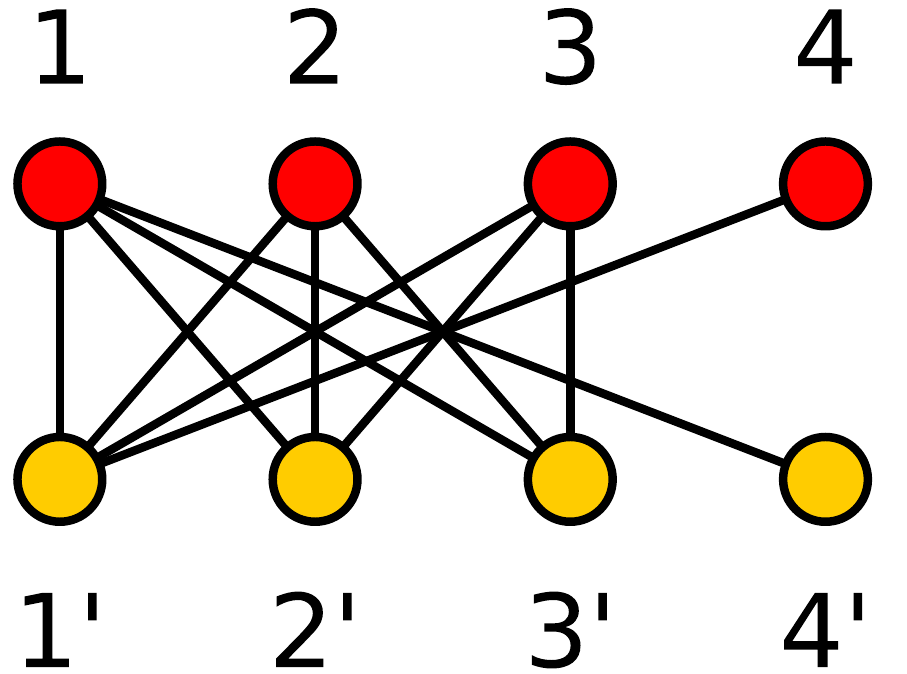}
\caption{Example for a threshold graph}
\end{subfigure}
~
\begin{subfigure}[b]{0.47\textwidth}
\centering
$B:=\left(\begin{matrix}
1&1&1&1\\
1&1&1&0\\
1&1&1&0\\
1&0&0&0\\
\end{matrix}\right)$
\caption{Corresponding biadjacency matrix}
\end{subfigure}
\caption{Example for a threshold graph}
\label{figure:thresholdGraph}
\end{figure}

The number of perfect matchings in such a graph can easily be counted by ordering the row sums $r_i$ for row $i$ in non-decreasing order, i.e., $r_1\geq r_2 \geq \dots \geq r_{n-1}\geq r_n$. Then $|M(G)|=r_n\cdot (r_{n-1}-1) \cdots (r_2-(n-2))(r_1-(n-1)).$ This formula was given by Brualdi and Ryser  \cite[Corollary 7.2.6.]{brualdi:91} and can easily be extended to an efficient approach for sampling. In each row $r_i$ - starting with~$r_n$ to~$r_1$ - we choose a column index $j$ with $B_{ij} = 1$ uniformly at random which was not used in a previous row $r_{i'}$ with~$i' > i$. We remove index $j$ in the set of all column indices and go on with the next row. Clearly, this simple approach is an exact uniform sampler for perfect matchings in threshold graphs. We consider the class of \emph{odd triangle threshold graphs} with $(n \times n)$-bi-adjacency matrices~$A$ with odd~$n \in \mathbb{N}$ and
 \[
A_{ij}=\begin{cases}
  1,  & \text{for }i\in \mathbb{N}_n \text{ and } j \in \{1,\dots,n-(i-1)\}\\
  0 & \text{else }
\end{cases}.
\]

\begin{proposition}\label{Proposition:thresholdGraph}
Broder's chain is not rapidly mixing for the class of odd triangle threshold graphs.
\end{proposition}

\begin{proof}
We consider the~$(n \times n)$ bi-adjacency matrix~$A$ for an odd triangle threshold graph $G = (U \cup V,E).$ We identify the~$i$-th row of~$A$  with the labeled vertex~$u_i \in U$ and the~$j$-th column with the labeled vertex~$v_j\in V.$ Clearly, this graph possesses exactly one unique perfect matching. Moreover, we find~$|N_{u_n v_n}|=2^{n-2}.$ We partition the set of near-perfect matchings in~$N_{u_n v_n}^1,$ the subset of all near-perfect matchings of~$N_{u_n v_n}$ containing the edge~$e = \{u_{n'},v_{n'}\}$ and~$N_{u_n,v_n}^2 := N_{u_n v_n} \setminus N_{u_n v_n}^1$, the subset containing the remaining near-perfect matchings, where~$n' := \left\lceil\frac{n}{2}\right\rceil$. It is easy to see that~$|N_{u_n v_n}^1|=2^{n-3}$ and so~$|N_{u_n,v_n}^2|=2^{n-3}.$\\
If we consider an arbitrary simple path~$P_{MM'}$ from~$M \in N_{u_n v_n}^1$ to~$M' \in N_{u_n v_n}^2$ in the corresponding state graph~$\Gamma$, such a path contains always a matching of the set~$S:=(\bigcup_{i \in \mathbb{N}_n}N_{u_{n'},v_i}) \cup (\bigcup_{i \in \mathbb{N}_n}N_{u_i,v_{n'}}),$ because~$M$ contains edge~$e$,~$M'$ not, and Broder's chain can only ``delete'' this edge by using a matching of~$S$. Moreover, we find that 
\[ 0= \left| N_{u_{n'},v_1} \right| = \left| N_{u_{n'},v_2} \right| = \dots <  \left| N_{u_{n'},v_{n'}} \right| = \left| N_{u_{n'},v_{n'+1}} \right| \leq \ldots \leq \left| N_{u_{n'},v_{n}} \right| = 2^{n'-2}. \]
(This can be seen by considering the bi-adjacency matrix~$A$.) For symmetry reasons we find that $|N_{u_{i},v_{n'}}|\leq2^{n'-2}$ for all~$i \in \mathbb{N}_n$. Hence,~$|S| \leq 2n(2^{n'-2}) <  n2^{n'}.$ If we now consider an arbitrary set~$\mathcal{P}'$ of simple paths between all vertex pairs~$(M,M')$ of~$N_{uv}^1\times N_{uv}^2$ (exactly one between each pair), then we find that~$|\mathcal{P}'|\geq 2^{2n-6}$ paths make use of a matching in~$S$. Each matching in~$S$ is adjacent to~$m$ arcs in~$\Psi$ (some of them can be loops). Hence, there exists an arc in~$\Gamma$ which has to be used from at least~$\left(2^{2n-6} \right)/ \left(mn2^{n'} \right)= \left(2^{\left\lfloor\frac{3n}{2}\right\rfloor-6} \right)/\left(mn \right)$ paths of~$\mathcal{P}'$. 
With equation~(\ref{eqn:ContributionMaximumLoading}) and since $|N_{u_nv_n}|\geq |N_{u_iv_j}|$ for $i,j \in \mathbb{N}_n$ we find
\[
\rho_1(\mathcal{P})\geq \rho(\mathcal{P}') \geq \frac{2^{\left\lfloor\frac{3n}{2}\right\rfloor-6}}{mn} \cdot \frac{m}{|\Omega|} \geq \frac{2^{\left\lfloor\frac{3n}{2}\right\rfloor-6}}{n|N_{u_n,v_n}|\left( n^2+1 \right)} = \frac{2^{\left\lfloor\frac{n}{2}\right\rfloor-4}}{n\left( n^2+1 \right)}
\]
for an arbitrary set~$\mathcal{P}$ of paths in~$\Gamma.$ Together with Proposition~\ref{Proposition:lowerBoundMixingtime} we find that the class of odd triangle threshold graphs is not rapidly mixing.\qed
\end{proof}

Note that we use similar techniques for proving it as for the hexagon graph class. More technical details but the same ideas will be used for showing that the two graph classes in Figure~\ref{figure:regularGraph} and Figure~\ref{figure:regularWithLadder} do also not possess a polynomial mixing time. The graph class in Figure~\ref{figure:regularGraph} is regular and planar but not Hamiltonian. The graph class in Figure~\ref{figure:regularWithLadder} is Hamiltonian and regular. Both classes can be scaled by inserting additional blocks to each of the three tracks and additional rungs in case of the graph class in Figure~\ref{figure:regularWithLadder}. Both graph classes possess an exponential number of perfect matchings and the fraction~$\frac{|N(G)|}{|M(G)|}$ is not bounded polynomially. The proof is again based on the idea to partition the set~$N_{uv}$ in two subsets which differ in the setting of perfect matchings in a middle block~$b_{k'}$ of the upper  subgraph starting with edge $\{u,v_1\}$ and ending with edge $\{u_k,v\}$ or $\{u'_l,v\},$ respectively.

\begin{figure}[tbp]
\centering
\includegraphics[width=8cm]{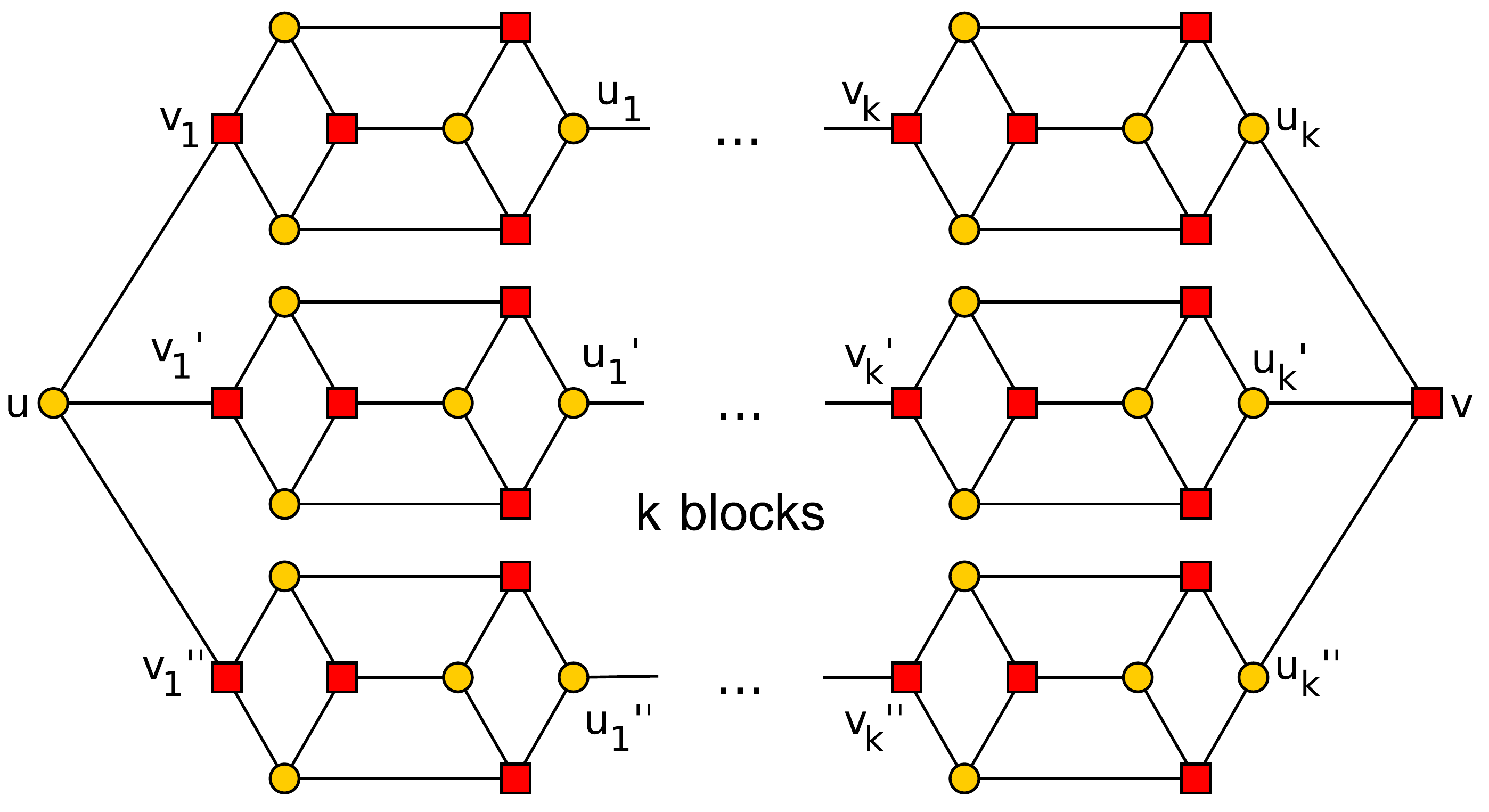}
\caption{Regular bipartite graph class which is not Hamiltonian and not rapidly mixing}
\label{figure:regularGraph}
\end{figure}

\begin{figure}[tbp]
\centering
\includegraphics[width=8cm]{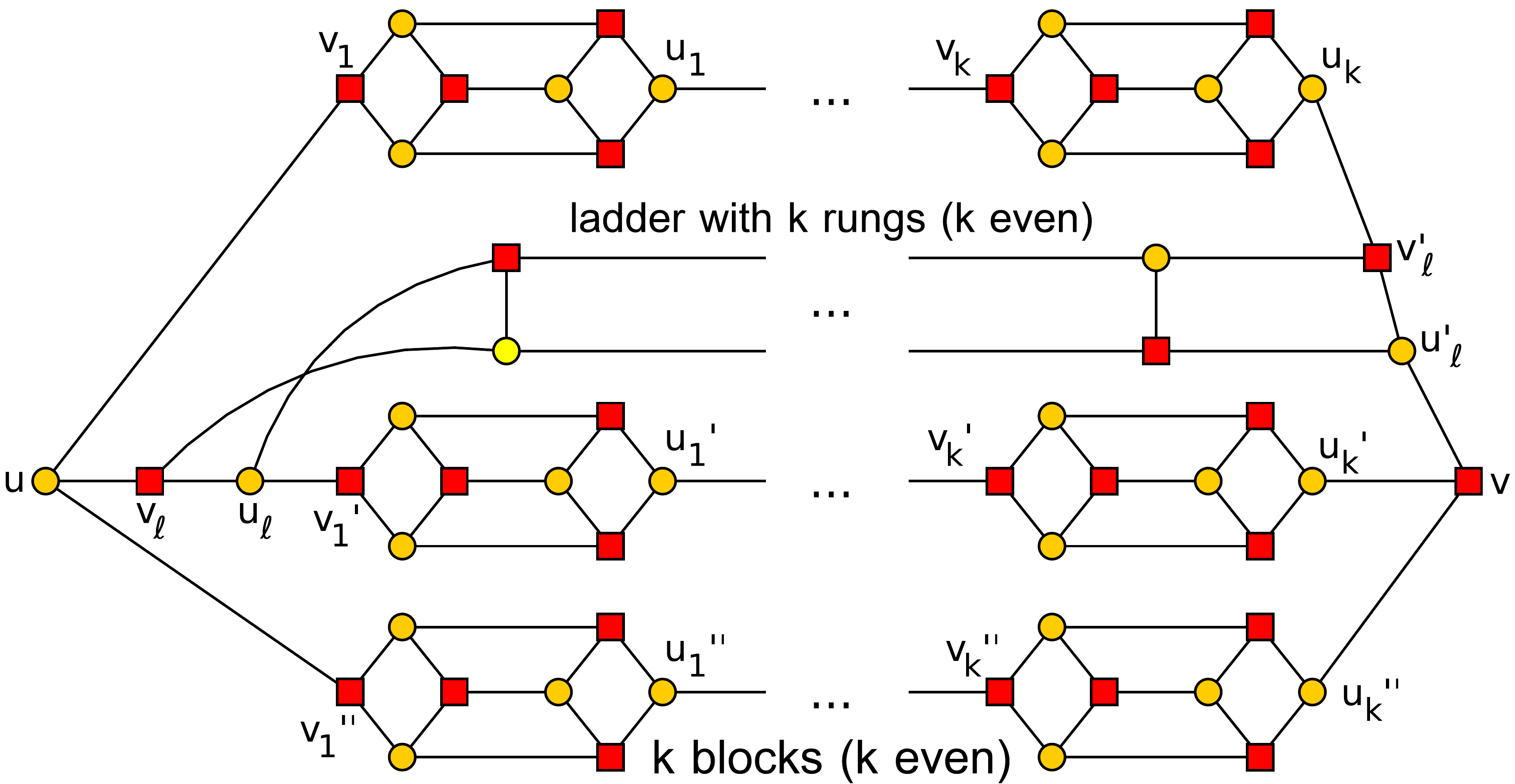}
\caption{Regular bipartite graph class which is Hamiltonian and not rapidly mixing}
\label{figure:regularWithLadder}
\end{figure}

\begin{proposition}\label{Proposition:regularGraph}
Broder's chain is not rapidly mixing for graphs in Figure~\ref{figure:regularGraph}.
\end{proposition}

\begin{proof}
The regular bipartite graph~$G=(U \cup V,E)$ in Figure~\ref{Proposition:regularGraph} consists of~$2n=24k+2$ vertices and~$3k$ blocks. The vertices~$u$ and~$v$ are connected via three rows of blocks. We consider each of the rows of~$G$ as a subgraph~$G_i$ where~$i \in \{1,2,3\}$. Each subgraph ends with the edges~$(u,v_i)$ and~$(u_k,v)$, respectively ($(u,v_i')$ and~$(u_k',v)$) or ($(u,v_i'')$ and~$(u_k'',v)$). We start with the computation of~$|M(G)|$ and~$|N_{uv}|.$ Let us consider one block in this graph. It possesses~$6$ possible perfect matchings~$M_i$,~$i \in \{1,\dots,6\}$. For a near-perfect matching in~$N_{uv}$ all blocks have one of the six possible perfect matchings~$M_i$. Hence,~$|N_{uv}|=6^{3k}.$ We consider the set of perfect matchings. Note that each perfect matching in~$G$ consists of exactly one perfect matching in a~$G_i$. For simplicity we set~$i:=1$. In the other two subgraphs~$G_2$ and~$G_3$ we have near-perfect~$N_{uv}$-matchings. Furthermore, each block in~$G_1$ possesses a near-perfect~$N_{v_j,u_j}$-matching ($j \in \{1,\dots,k\}$) where three possibilities per block exist. We get~$|M(G)|=6^{2k}3^{k+1}.$ Furthermore, we find that~$|N_{uv}|\geq |N_{u^*,v^*}|$ for all~$u^* \in U, v^* \in V,$ because the number of blocks where~$6$ possibilities for a perfect matching occur is smaller in a near-perfect matching of~$N_{u^*v^*}$ where at least one statement:~$u^* \neq u$ or~$v^* \neq v$ is true. We partition the set~$N_{uv}$ in the subset of matchings~$N_{uv}^1$ containing in the~$k' = \left\lceil k/2 \right\rceil$th block of~$G_1$ one of the perfect matchings~$M_1,M_2,$ or~$M_3$. The remaining set which contains all near-perfect matchings of~$N_{uv}$ with one of the other three possible perfect matchings~$M_4,M_5,M_6$ we denote by~$N_{uv}^2.$ Hence,~$|N_{uv}^1|=|N_{uv}^2|=3 \cdot 6^{3k-1}$. We consider a simple path in the corresponding state graph~$\Gamma$ between an arbitrary vertex pair~$(M,M')\in N_{uv}^1\times N_{uv}^2$. Such a path has to change a perfect matching in the~$k'$th block. Using Broder's chain this is only possible by using a path in~$\Gamma$ which has a vertex~$N_{u^*v^*}$ such that~$u^*$ or~$v^*$ is a vertex in the~$k'$th block of~$G_1$. We denote the set of vertices in the~$k'$th block by~$U'$ or~$V'$, respectively. Then at least one vertex of the set~$S:=(\bigcup_{u^* \in U,v'\in V'} N_{u^*,v'}) \cup (\bigcup_{u'\in U',v^* \in V} N_{u'v^*})$ occurs in a simple path between~$M$ and~$M'$. We consider the cardinality of~$S$ in upper bounding~$N_{u^*v'}$ and distinguish between three possible cases:
\begin{description}
\item[case~1:]~$u^* \in U'$ (see Figure~\ref{figure:CaseOne}). Note that this case is only possible when~$u_{k'}$ and~$v_{k'}$ are both unmatched in the~$k'$th induced block or both are matched in this induced block. Otherwise we cannot find a suitable near-perfect matching because such a setting enforces that exactly one of the edges~$(u,v_1)$ and~$(u_k,v)$ is matched, which is not possible. Assume (a) that~$u^*=u_{k'}$ and~$v'=v_{k'},$ then all blocks in~$G_1$ (without block~$b_{k'}$) contain one of the six perfect matchings~$M_i$ and~$(u,v_1)$ and~$(u_k,v)$ are unmatched. This enforces that either~$G_2$ or~$G_3$ contains a perfect matching. Hence,~$|N_{u^*,v'}|\leq 6^{2k}3^k \cdot 2.$ If~$u^*\neq u_{k'}$ or~$v'\neq v_{k'}$ then edge pair~$(u_{k'-1},v_{k'}),(u_{k'},v_{k'+1})$ is either (b) matched or (c) not. In case (b) we get that~$(u,v_1)$ and~$(u_k,v)$ are matched and all blocks in~$G_1$ (without~$b_{k'}$) possess one of three near-perfect~$N_{v_j,u_j}$-matchings. The~$k'$th block has at most two matchings. In case (c),~$(u,v_1)$ and~$(u_k,v)$ are unmatched and all blocks in~$G_1$ (without~$k'$) possess one of six perfect matchings~$M_i$. This enforces that~$G_2$ or~$G_3$ contain a perfect matching. Case (b) and (c) can occur together when~$u^*\neq u_{k'}$ and~$v'\neq v_{k'}$ and yield~$|N_{u^*,v'}|\leq 3 \cdot 6^{2k}3^k.$ In summary, case~1 leads to the bound~$|N_{u^*,v'}|\leq 3 \cdot 6^{2k}3^k.$
\begin{figure}[tbp]
\centering
\includegraphics[width=16cm]{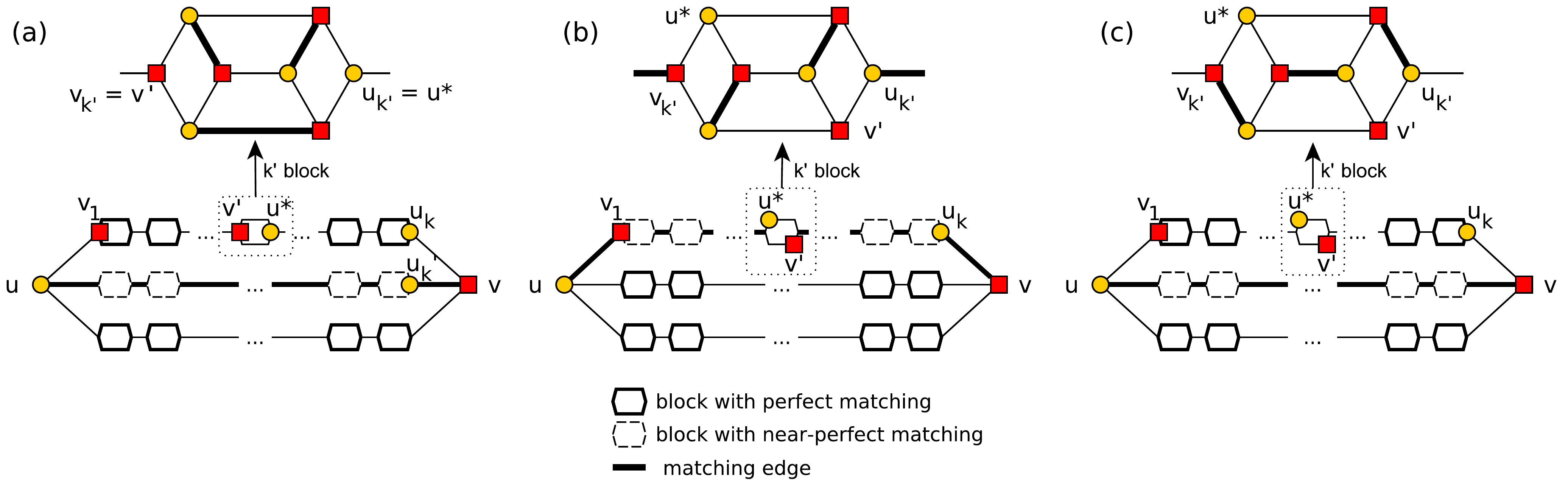}
\caption{Schematic figure of case distinctions in case~1}
\label{figure:CaseOne}
\end{figure}
\item[case~2:]~$u^* \in U(G_1)\setminus (U'\cup \{u\})$ (see Figure~\ref{figure:CaseTwo}). Since the number of vertices in each block is even, this case is only possible if exactly one of the edges~$(u_{k'-1},v_{k'})$ and~$(u_{k'},v_{k'+1})$ is in~$N_{u'v^*}$. Since~$v'$ is not matched, it remains the case that~$(u_{k'},v_{k'+1})$ is matched. We have to distinguish the cases that~$u^*$ is in block~$b_{k^*}$ with (a)~$k^*<k'$ or (b)~$k^*>k'.$ If~$u^*$ occurs in block~$k^*$ we find for the adjacent edges of this block that~$(u_{k^*-1},v_{k^*})$ is matched and~$(u_{k^*},v_{k^*+1})$ is not. For~$k^*<k'$ we get that~$(u,v_1)$ and~$(u_k,v)$ are both matched. Especially all blocks~$b_{k'+1},\dots,b_k$ possess one of three~$N_{v_j,u_j}$-near-perfect matchings. Hence,~$|N_{u^*v'}|\leq 3^{k-k'}6^{2k+k'}$, because the blocks between~$k^*$ and~$k'$ have~$6$ possibilities for a perfect matching~$M_i$. For~$k^*>k'$ we find that~$(u,v_1)$ and~$(u_k,v)$ are both unmatched. Therefore, either~$G_2$ or~$G_3$ has a perfect matching. Hence,~$|N_{u^*v'}|\leq 2 \cdot 3^{k}6^{2k}$. In summary case~2 leads to~$|N_{u^*v'}|\leq 3^{k-k'}6^{2k+k'}$.
\begin{figure}[tbp]
\centering
\includegraphics[width=14cm]{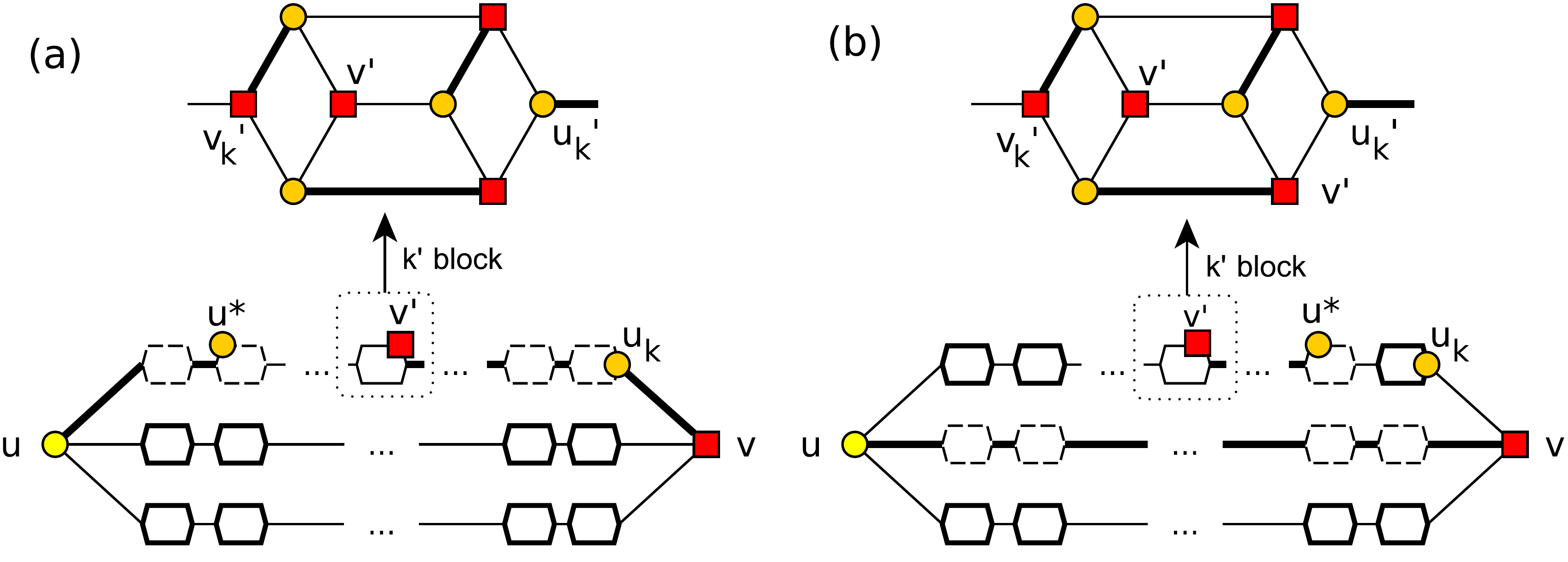}
\caption{Schematic figure of case distinctions in case~2}
\label{figure:CaseTwo}
\end{figure}
\item[case~3:]~$u^* \in U(G_2) \cup U(G_3)$ (see Figure~\ref{figure:CaseThree}). For simplicity we assume that~$u^* \in U(G_2).$ If (a)~$u^*$ occurs in the~$k^*$th block of~$G_2,$ then this enforces that the edge~$(u_{k^*-1}',v_{k^*}')$ is matched and the edge~$(u_{k^*}',v_{k^*+1}')$ is not. Therefore, all blocks~$b_1,\dots,b_{k^*-1}$ in~$G_2$ have a near-perfect~$N_{v_j,u_j}$-matching and all blocks~$b_{k^*+1},\dots,b_k$ a perfect matching~$M_i$. Graph~$G_3$ contains  for each block one of six perfect matchings. The case (b)~$u=u^*$ enforces that ~$(u,v_1)$ is unmatched and~$(u_k,v)$ is matched. It follows~$|N_{u^*v'}|\leq 3^{k-k'}6^{2k+k'}$ for cases (a) and (b). 
\begin{figure}[tbp]
\centering
\includegraphics[width=14cm]{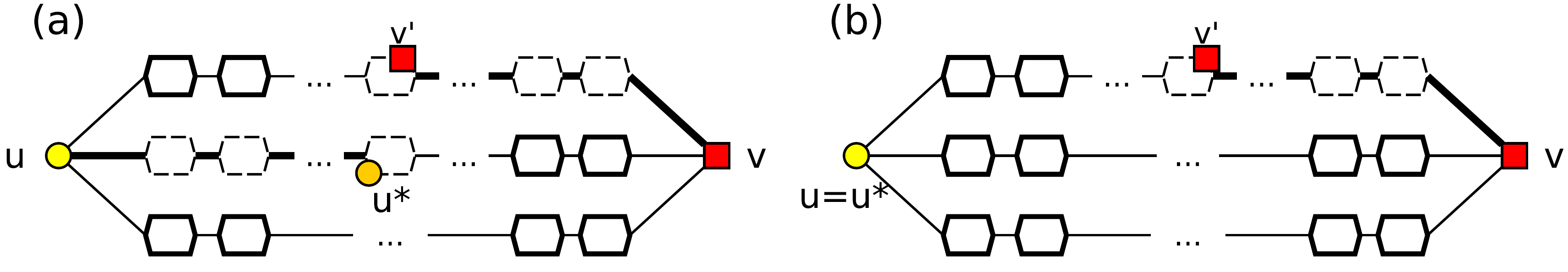}
\caption{Schematic figure of case distinctions in case~3}
\label{figure:CaseThree}
\end{figure}
\end{description}

We get the upper bound~$|S|\leq 8n 3^{k-k'}6^{2k+k'}$ if we take into account that there exist~$4n$ sets~$N_{u^*v'}$ and all sets~$N_{u'v^*}$ can be handled analogously. If we consider a set~$\mathcal{P}'$ of simple paths between all vertex pairs~$(M,M')$ of~$N_{uv}^1\times N_{uv}^2$ (exactly one between each pair), then we find that~$|\mathcal{P}'|=6^{6k-2}$ paths use an arc in~$\Gamma$ which is adjacent to vertex set~$S$. Note that each vertex in~$S$ is adjacent to at most~$m$ arcs (loops are possible). Hence, there exists an arc~$a$ in~$\Gamma$ which has to be used from at least~$( 6^{6k-2} ) /  ( 8nm 3^{k-k'}6^{2k+k'} ) = ( 6^{3k-2}2^{\left\lfloor\frac{k}{2}\right\rfloor} ) / ( 8mn )$ paths of~$\mathcal{P}'$.
With equation~(\ref{eqn:ContributionMaximumLoading}) we get for the maximum loading of an arbitrary set~$\mathcal{P}$ of paths in~$\Gamma,$
\footnotesize
$$\rho_1(\mathcal{P})\geq \rho_1(\mathcal{P}') \geq \frac{6^{3k-2}2^{\left\lfloor\frac{k}{2}\right\rfloor}}{8mn} \cdot \frac{m}{|\Omega|} =  \frac{6^{3k-2}2^{\left\lfloor\frac{k}{2}\right\rfloor}}{8n|N_{uv}|(n^2+1)} \geq \frac{2^{\left\lfloor\frac{k}{2}\right\rfloor}}{288n(n^2+1)}=\frac{2^{\left\lfloor\frac{n-1}{24}\right\rfloor}}{288n(n^2+1)}.$$
\normalsize
Together with Proposition~\ref{Proposition:lowerBoundMixingtime} we find that the graph class does not possess a polynomial mixing time.\qed
\end{proof}

\begin{proposition}\label{Proposition:regularWithLadder}
Broder's chain is not rapidly mixing for graphs in Figure~\ref{figure:regularWithLadder}.
\end{proposition}

\begin{proof}
The regular bipartite graph~$G=(U \cup V,E)$ in Figure~\ref{figure:regularWithLadder} consists of~$2n=26k+6$ vertices,~$3k$ blocks and a ladder~$L_{k+2}$ with~$k+2$ rungs. Clearly, this is the graph of Figure~\ref{figure:regularGraph} with an additional ladder. Note that~$k$ has to be even to get a bipartite graph. We consider the upper part of~$G$ as subgraph~$G_1$, starting and ending with the edges~$(u,v_1), (u,v_l),(u_{k}',v)$. The lower part of~$G$ is called~$G_2$ with~$(u,v_1''),(u_k'',v)$ on its ends. We start with the computation of~$|M(G)|$ and~$|N_{uv}(G)|.$ Let us consider a block in this graph. It possesses~$6$ possible perfect matchings~$M_i$,~$i \in \{1,\dots,6\}$.
We find that~$|N_{uv}(G)|=|N_{uv}(G_1)||N_{uv}(G_2)|.$ Note that when we consider a matching~$M \in N_{uv}(G)$ and the edges~$\{u,v_1\}$ and~$\{u_{k}',v\}$ do not belong to~$M$, this implies that~$\{u_l,v_{1}'\}$ and~$\{u_k,v_{l}'\}$ also do not belong to~$M$. Hence, the ladder~$L_{k+2}$ contains a perfect matching in each~$N_{uv}(G).$ We find that~$|N_{uv}(G_1)|=6^{2k}|M(L_{k+2})|$. The number of perfect matchings~$|M(L_k)|$ in a ladder~$L_k$ with~$k$ rungs can be computed in the following way. $|M(L_1)| = 1$, $|M(L_2)|=2$, $|M(L_3)| = |M(L_1)| + |M(L_2)|$ and in general, $|M(L_k)|  = |M(L_{k-1})| + |M(L_{k-2})|$. Hence,~$|M(L_k)|$ with~$k \in \mathbb{N}$ is the famous \emph{Fibonacci sequence} and~$|M(L_{k-1})|=F_k$ denotes the~$k$th Fibonacci number with~$F_k=\frac{1}{\sqrt{5}}(\Phi^k-\hat{\Phi})$,~$\Phi=\frac{1+\sqrt{5}}{2}$ is the \emph{golden ratio} and~$\hat{\Phi}=\frac{1-\sqrt{5}}{2}$. (Note that a ladder with~$0$ rungs has one perfect matching and a ladder with~$1$ rung, too.) We find~$|N_{uv}(G)|=6^{3k}F_{k+3}$. Let us consider the set of perfect matchings~$M(G)$ in this graph. Clearly, we have~$|M(G)|=|N_{uv}(G_1)||M(G_2)|+|M(G_1)||N_{uv}(G_2)|=6^{2k}F_{k+3}3^{k}+6^k|M(G_1)|$. Let us consider the set~$M(G_1).$ Note that~$(u,v_1) \in M$ enforces that edge~$(u_k,v_{l}')$ is also contained in this matching~$M$. The same is true for the edges~$(u_{k}',v)$ and~$(u_l,v_{1}').$ We distinguish between four cases, (a)~$(u,v_1),(u_{k}',v)$ are matched edges, (b)~$(u,v_1),(u_{l}',v)$ are matched edges, (c)~$(u,v_l),(u_l',v)$ are matched, and (d)~$(u,v_l),(u_k',v)$ are matched. Consider Figure~\ref{figure:LadderCases}. Case (b) and (d) are analogous for symmetry reasons. 
\begin{figure}[tbp]
\centering
\includegraphics[width=14cm]{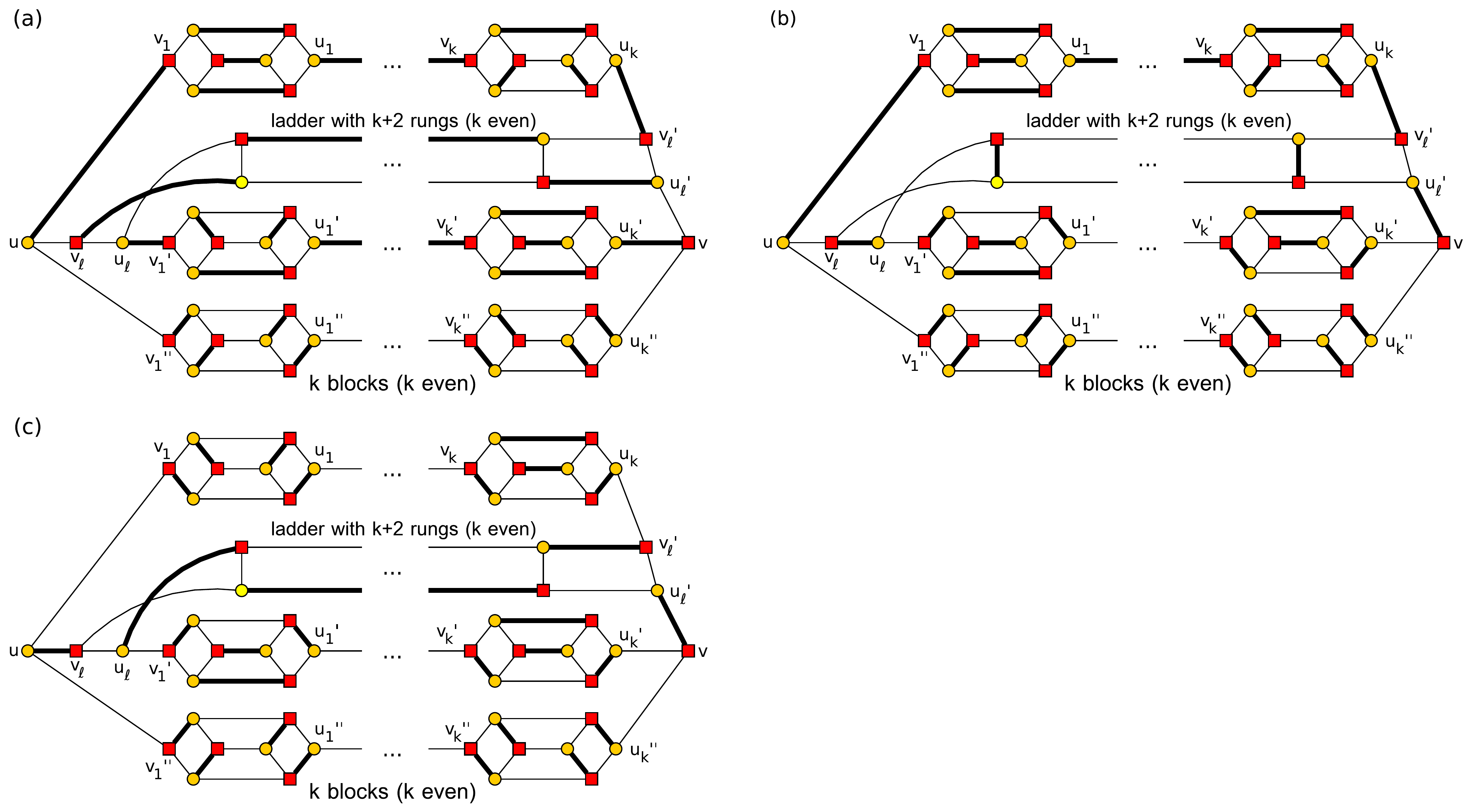}
\caption{Distinct perfect matching types in~$G_1$ of the regular graph class with ladder}
\label{figure:LadderCases}
\end{figure}
Case (a) enforces that~$(u_k,v_{l}')$ and~$(u_l,v_{1}')$ are also matching edges. It only remains one near-perfect~$N_{u_lv_{l'}}$-matching for~$L_{k+2}$. For all blocks~$b_i$ with~$i = \{1, \ldots, k \}$ in~$G_1$ we find~$3$ possibilities for a near-perfect~$N_{u_i,v_i}$-matching. We find~$3^{2k}$ perfect matchings. Case (b) enforces that~$(u_{k},v_{l}')$ is a matched edge. It remains to find a perfect matching in a sub-ladder~$L_{k+1}$ with~$k+1$ rungs. We find for each block~$b_1,\dots,b_k$ in~$G_1$ three possibilities for a near-perfect~$N_{u_i,v_i}$-matching ($i \in \{1, \ldots k\}$) and for each block~$b_1',\dots,b_k'$ six possibilities for a perfect matching. We get~$F_{k+2}3^k6^k$ perfect matchings. Case (c) enforces that edges~$(u_l,v_1'),(u_k,v_l')$ are not matched. We find for each block in~$G_1$ six possibilities for a perfect matching and one~$N_{u_l',v_l}$-near perfect matching in~$L_{k+1}$. We find~$6^{2k}$ perfect matchings in~$G_1$. In summary we get for all four cases,~$|M(G_1)|=3^{2k}+2F_{k+2}3^k6^k+6^{2k}$ and so 
$$|M(G)|=6^{2k}3^{k}(F_{k+3}+2F_{k+2})+6^k3^{2k}(1+2^{2k}) \leq 6^{2k} 3^{k}(2 F_{k+4})+6^k 3^{2k}(1+2^{2k}) .$$
Since~$2 >\Phi >1$ we find 
\begin{equation}
F_k = \frac{\Phi^k}{\sqrt{5}}-\frac{1-\sqrt{5}}{2 \sqrt{5}}\leq \frac{2^{k+1}}{\sqrt{5}} \leq 2^{k+1}.
\end{equation}
and 
\begin{equation}
F_k = \frac{\Phi^k}{\sqrt{5}}-\frac{1-\sqrt{5}}{2 \sqrt{5}}\geq \frac{\Phi^k}{\sqrt{5}}=\frac{1}{\sqrt{5}}\left(\frac{1+\sqrt{5}}{2}\right)^k.
\label{eqn:Fibonacci}
\end{equation}
Furthermore, we get~$\frac{N_{uv}(G)}{M(G)}\geq \frac{6^{3k}F_{k+3}}{6^{2k}3^{k}2F_{k+4}+6^k3^{2k}(1+2^{2k})} \geq \frac{6^{3k}F_{k+3}}{6^{2k} 3^k 2^{k+6}+6^{k}3^{2k}(2^{2k+1})}=\frac{6^{3k}F_{k+3}}{6^{3k}2^6 +6^{3k}2}=\frac{F_{k+3}}{66}$ is exponentially large.\\
We partition the set~$N_{uv}(G)$ in the subset of matchings~$N_{uv}^1$ containing in the~$k'$th block ($k'=\frac{k}{2}$) of~$G_1$ one of the perfect matchings~$M_1,M_2$ or~$M_3$. The remaining set which contains all near-perfect matchings of~$N_{uv}$ with one of the other three possible perfect matchings~$M_4,M_5,M_6$ we denote by~$N_{uv}^2.$ Hence,~$|N_{uv}^1|=|N_{uv}^2|=6^{3k-1}F_{k+3}\cdot 3$. We consider an arbitrary simple path in the corresponding state graph~$\Gamma$ between an arbitrary vertex pair~$(M,M')\in N_{uv}^1\times N_{uv}^2$. Such a path has to change a perfect matching in the~$k'$th block. Using Broder's chain this is only possible by using a path in~$\Gamma$ which has a vertex~$N_{u^*v^*}$ such that~$u^*$ or~$v^*$ is a vertex in the~$k'$th block of~$G_1$. We denote the set of vertices in the~$k'$th block by~$U'$ or~$V'$, respectively. Then at least one vertex of set~$S:=(\bigcup_{u^* \in V,v'\in V'} N_{u^*,v'}) \cup (\bigcup_{v^* \in V,u'\in V'} N_{u'v^*})$ occurs in a simple path between~$M$ and~$M'$. We consider the cardinality of~$S$ in upper bounding~$|N_{u^*v'}|$ and distinguish between four possible cases:
\begin{description}
\item[case~1:]~$u^* \in U'$ (see Figure~\ref{figure:CaseOneLadder}). Note that this case is only possible when~$u_{k'}$ and~$v_{k'}$ are both unmatched in the~$k'$th induced block or both are matched in this induced block, because the number of vertices in a block is even. Assume~$u^*=u_{k'}$ and~$v'=v_{k'}.$ Then all blocks in~$G_1$ (without~$k'$) contain one of the six perfect matchings~$M_i$ and~$(u,v_1)$ and~$(u_k,v_{l}')$ are unmatched. This enforces that either (a) the edge pair~$(u,v_l),(u_{k'},v)$ is matched, (b) the edge pair~$(u,v_l),(u_{l'},v)$ is matched or (c)~$G_3$ contains a perfect matching. Case (a) enforces that each block~$b_1,\dots,b_k$ and~$b_1'',\dots, b_k''$ (without block~$b_{k'}$) possesses one of six possibilities for a perfect matching~$M_i.$ The blocks~$b_1',\dots,b_k'$ have three possibilities for a near-perfect~$N_{u_{k}',v_{k}'}$-matching. There are~$F_{k+2}$ perfect matchings in a ladder with~$k+1$ rungs. We get an upper bound of~$6^{2k}3^{k}F_{k+2}$ matchings. Case (b) enforces that all blocks (without~$b_{k'})$ have six possibilities for a perfect matching. The ladder possesses only one unique near-perfect matching. We get an upper bound of~$6^{3k}$ matchings. In case (c) we find an upper bound of~$6^{2k}3^{k}F_{k+3}$ matchings. If~$u^*\neq u_{k'}$ or~$v'\neq v_{k'}$ then either both edges~$(u_{k'-1},v_{k'}),(u_{k'},v_{k'+1})$ are unmatched or they are matched edges. In the first case we find the same three situations as in (a)-(c). In the second case, we find that~$(u,v_1)$ and~$(u_k,v_{l}')$ are matched and all blocks~$b_1,\dots,b_k$ (without~$b_{k'}$) possess one of three near-perfect~$N_{v_j,u_j}$-matching. This case has the subcase (d) that edge~$(u_{l}',v)$ is a matching edge or subcase (e) that edge~$(u_{k}',v)$ is a matching edge. In (d) we find~$3^k6^{2k}F_{k+2}$ matchings. In (e) we have~$3^{2k}6^k$ matchings, because the ladder has only one unique near-perfect matching. Since it can happen that edge pair~$(u_{k'-1},v_{k'}),(u_{k'},v_{k'+1})$ is matched in one matching of~$N_{u^*,v'}$ and is unmatched for another near-perfect matching in this set, we put all cases (a)--(e) together to upper bound~$|N_{u^*,v'}|$ in case~1 and get
$$|N_{u^*,v'}|\leq 2 \cdot 6^{2k}3^kF_{k+2}+6^{3k}+6^{2k}3^kF_{k+3}+3^{2k}6^k \leq 3 \cdot 3^k6^{2k}F_{k+3}+2 \cdot 6^{3k}.$$
\begin{figure}[tbp]
\centering
\includegraphics[width=14cm]{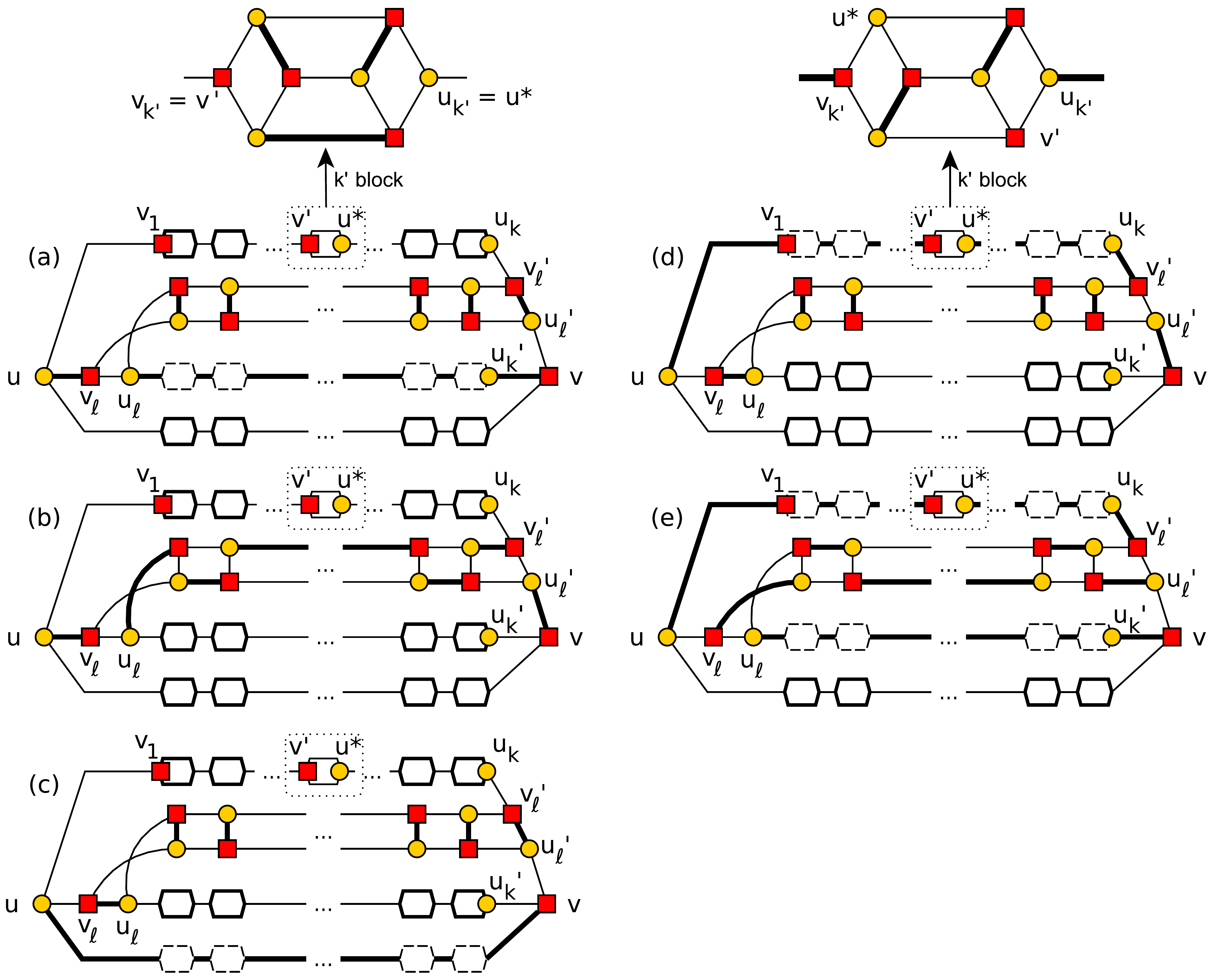}
\caption{Schematic figure of case distinctions for case~1}
\label{figure:CaseOneLadder}
\end{figure}
\item[case~2:]~$u^*\notin V'$ is in one block between~$b_1,\dots,b_k$ (see Figure~\ref{figure:CaseTwoLadder}). Since the number of vertices in each block is even, this case is only possible if exactly one of the edges~$(u_{k'-1},v_{k'})$ and~$(u_{k'},v_{k'+1})$ is matched. Since~$v'$ is not matched and the number of vertices in a block is even it remains the case that~$(u_{k'},v_{k'+1})$ is matched. We have to distinguish the cases that~$u^*$ is in block~$k^*<k'$ or~$k^*>k'$. If~$u^*$ occurs in block~$k^*$ we find that~$(u_{k^*-1},v_{k^*})$ is matched and~$(u_{k^*},v_{k^*+1})$ not. For~$k^*<k'$ we get that~$(u,v_1)$ and~$(u_k,v_{l}')$ are matched. Especially all blocks~$b_1,\dots,b_{k^*-1}$ and~$b_{k'+1},\dots,b_k$ possess three~$N_{v_j,u_j}$-near-perfect matchings. There can be two subcases. Either (a) edge~$(u_{l}',v)$ is matched or (b) edge~$(u_{k}',v)$ is matched. For (a) and (b) together we get~$|N_{u^*,v'}|\leq 6^{2k+k'}3^{k'}F_{k+2}+3^{k'+k}6^{k'+k} \leq 6^{2k+k'}3^{k'}F_{k+2}+6^{3k}~$. For~$k^*>k'$ we find that~$(u,v_1)$ and~$(u_k,v_{l}')$ are unmatched. Hence,~$(u,v_{l})$ is matched or~$(u,v_{1''})$ is matched. We distinguish for the first case between (c) that edge~$(u_{k}',v)$ is matched and (d) edge~$(u_{l}',v)$ is matched. For both cases together we find~$|N_{u^*,v'}|\leq 6^{2k}3^kF_{k+2}+6^{3k}$. For the second case (e) we find~$|N_{u^*,v'}|\leq 6^{2k}3^kF_{k+3}.$ In summary we get for case~2 that~$|N_{u^*,v'}|\leq 6^{2k+k'}3^{k'}F_{k+2}+6^{3k}$.
\begin{figure}[tbp]
\centering
\includegraphics[width=14cm]{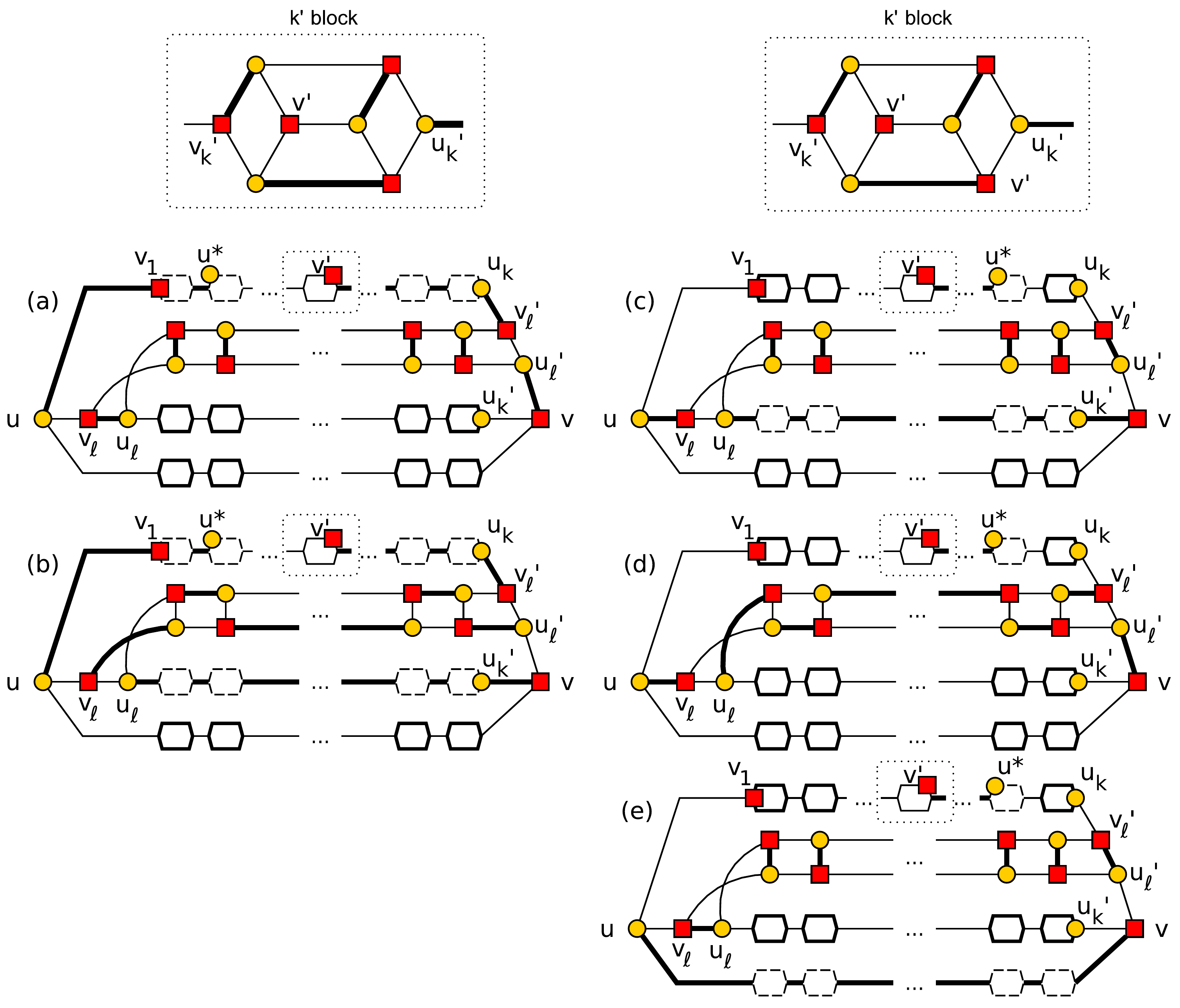}
\caption{Schematic figure of case distinctions in case~2}
\label{figure:CaseTwoLadder}
\end{figure}
\item[case~3:]~$u^*$ is in one block of~$b_{1}',\dots,b_k'$ or~$u^*=u$ (see Figure~\ref{figure:CaseThreeLadder}). If~$u^*$ occurs in block~$k^* \in \{b_1',\dots,b_k'\}$ then this enforces that the edge~$(u_{k^*-1}',v_{k^*}')$ is matched and the edge~$(u_{k^*}',v_{k^*+1}')$ not. Therefore, all blocks between~$b_1',\dots,b_{k^*-1}'$ have a near-perfect~$N_{v_j',u_j'}$-matching and all blocks between~$b_{k^*+1}',\dots,b_k'$ a perfect matching~$M_i$. The graph~$G_2$ contains either (a) one of the possible near-perfect~$N_{uv}$-matchings or (b) a possible perfect matching. In (a) there exist~$F_{k+1}$ possibilities for a perfect matching in the ladder with~$k$ rungs. Hence, we find for (a) and (b) together,~$|N_{u^*v'}|\leq 3^{k'}6^{2k+k'}F_{k+2}+3^{2k}6^k+3^{k+k'}6^{k+k'}$. The case~$u=u^*$ enforces that ~$(u,v_1)$ is unmatched and~$(u_k,v_{l}')$ is matched. In this scenario we have the cases (c), where~$(u_{l}',v)$ is matched or case (d), where~$(u_{k}',v)$ is matched. For (c) and (d) together we find~$|N_{u^*v'}|\leq 3^{k'}6^{2k+k'}F_{k+2}+3^k6^{2k} \leq 3^{k'}6^{2k+k'}F_{k+2}+6^{3k}$. In summary we get for case~3,~$|N_{u^*v'}|\leq 3^{k'}6^{2k+k'}F_{k+2}+6^{3k}$.
\begin{figure}[tbp]
\centering
\includegraphics[width=14cm]{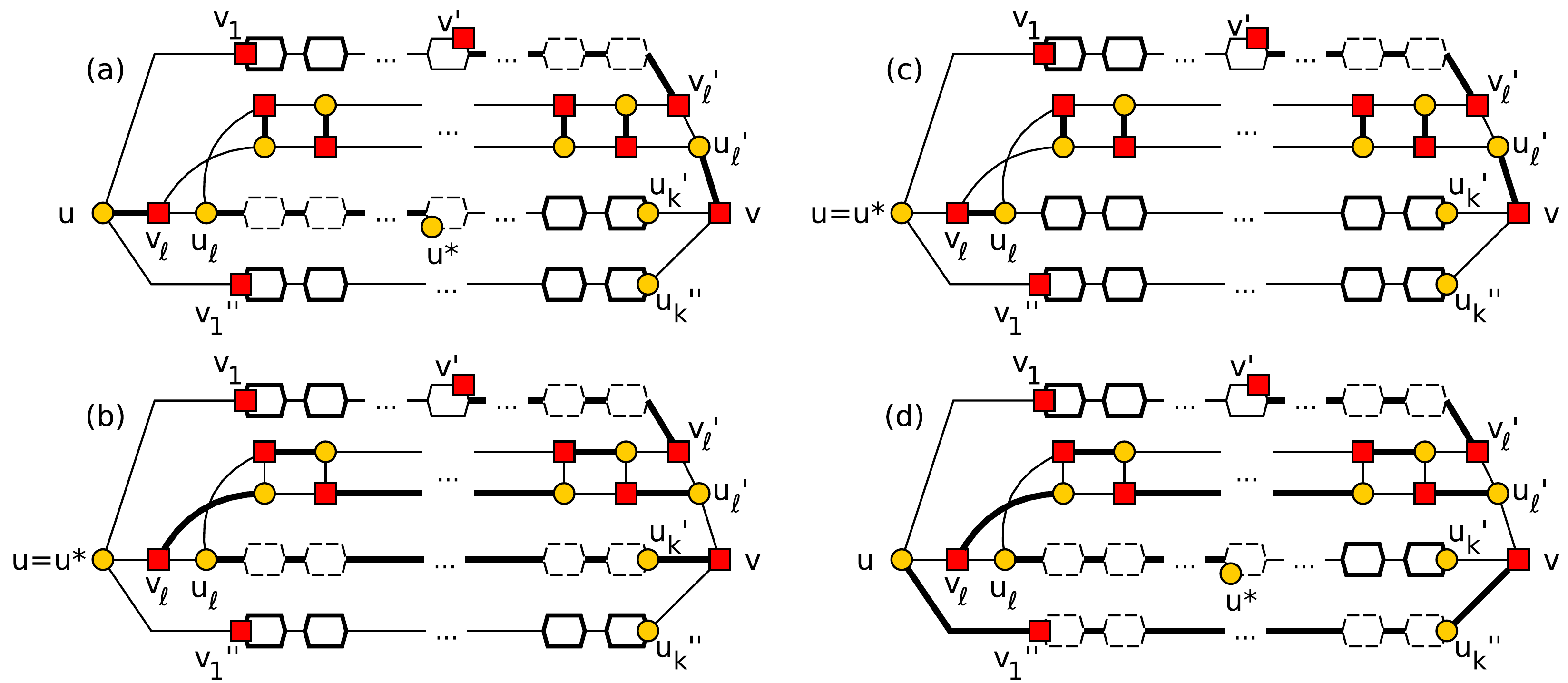}
\caption{Schematic figure of case distinctions in case~3}
\label{figure:CaseThreeLadder}
\end{figure}
\item[case~4:]~$u^*$ is in one block of~$b_1'',\dots,b_k''$ (see Figure~\ref{figure:CaseFourLadder}). This enforces that edge~$(u_{k^*-1}'',v_{k^*}'')$ is matched and edge~$(u_{k^*}'',v_{k^*+1}'')$ not. Therefore, all blocks between~$b_1'',\dots,b_{k^*-1}''$ have a near-perfect~$N_{v_j'',u_j''}$-matching and all blocks between~$b_{k^*+1}'',\dots,b_k''$ a perfect matching~$M_i$. This enforces that the edge~$(u,v_1'')$ is matched and the edge~$(u_{k}'',v)$ not. Two subcases occur, namely case (a), where edge~$(u_{l}',v)$ is matched and (b) where edge~$(u_{k}',v)$ is matched. For both cases together we get~$|N_{u^*v'}|\leq 3^{k'}6^{2k+k'}F_{k+2}+3^{k+k'}6^{k+k'} \leq 3^{k'}6^{2k+k'}F_{k+2}+6^{3k}$.
\begin{figure}[tbp]
\centering
\includegraphics[width=8cm]{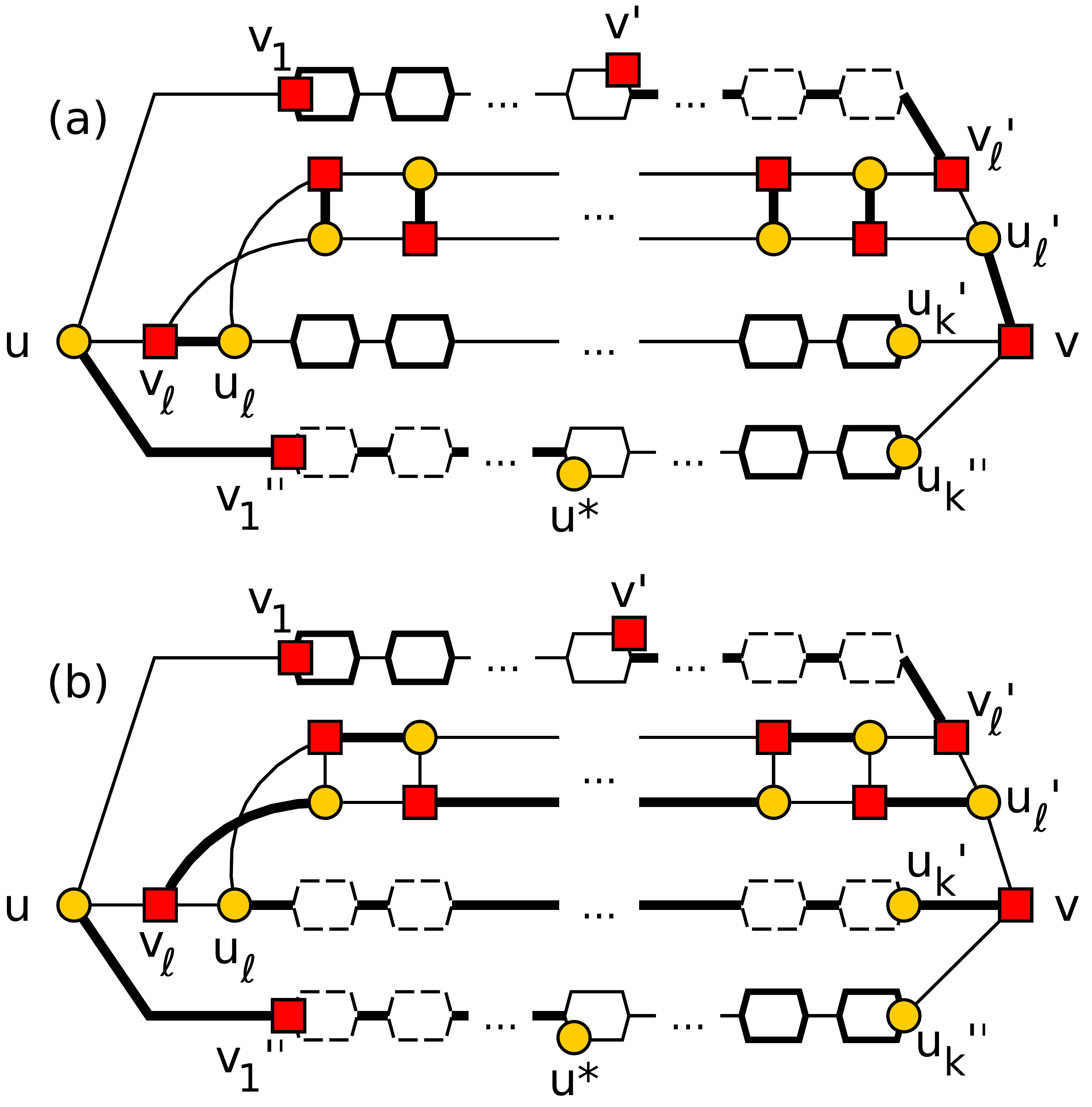}
\caption{Schematic figure of case distinctions in case~4}
\label{figure:CaseFourLadder}
\end{figure}
\item[case~5:]~$u^*$ in~$L_{k+2}$ (see Figure~\ref{figure:CaseFiveLadder}). If~$u^*=u_l$ then this enforces that either (a) edge~$(u,v_{l})$ is matched or (b) edge~$(u,v_{1}'')$ is matched. Considering the pictures we see that we get for both cases together,~$|N_{u^*v'}|\leq 3^{k'}6^{2k+k'}F_{k+1}+3^{k+k'}6^{k+k'} \leq 2 \cdot 3^{k'}6^{2k+k'}F_{k+1}$. For~$u^*\neq u_l$ we distinguish between the cases (c) with matched edge~$(u_{l}',v)$, (d) matched edge~$(u_{k}',v),$ and (e)  matched edge~$(u_{k}'',v)$. We find for all three cases together~$|N_{u^*v'}|\leq 3^{k'}6^{2k+k'}F_{k}+3^{k+k'}6^{2k+k'} F_k+ 3^{k+k'}6^{k+k'} F_k\leq 3 \cdot 3^{k'}6^{2k+k'}F_{k}$. In summary we find for case~5,~$|N_{u^*v'}|\leq 3 \cdot  3^{k'}6^{2k+k'}F_{k}$.
\begin{figure}[tbp]
\centering
\includegraphics[width=14cm]{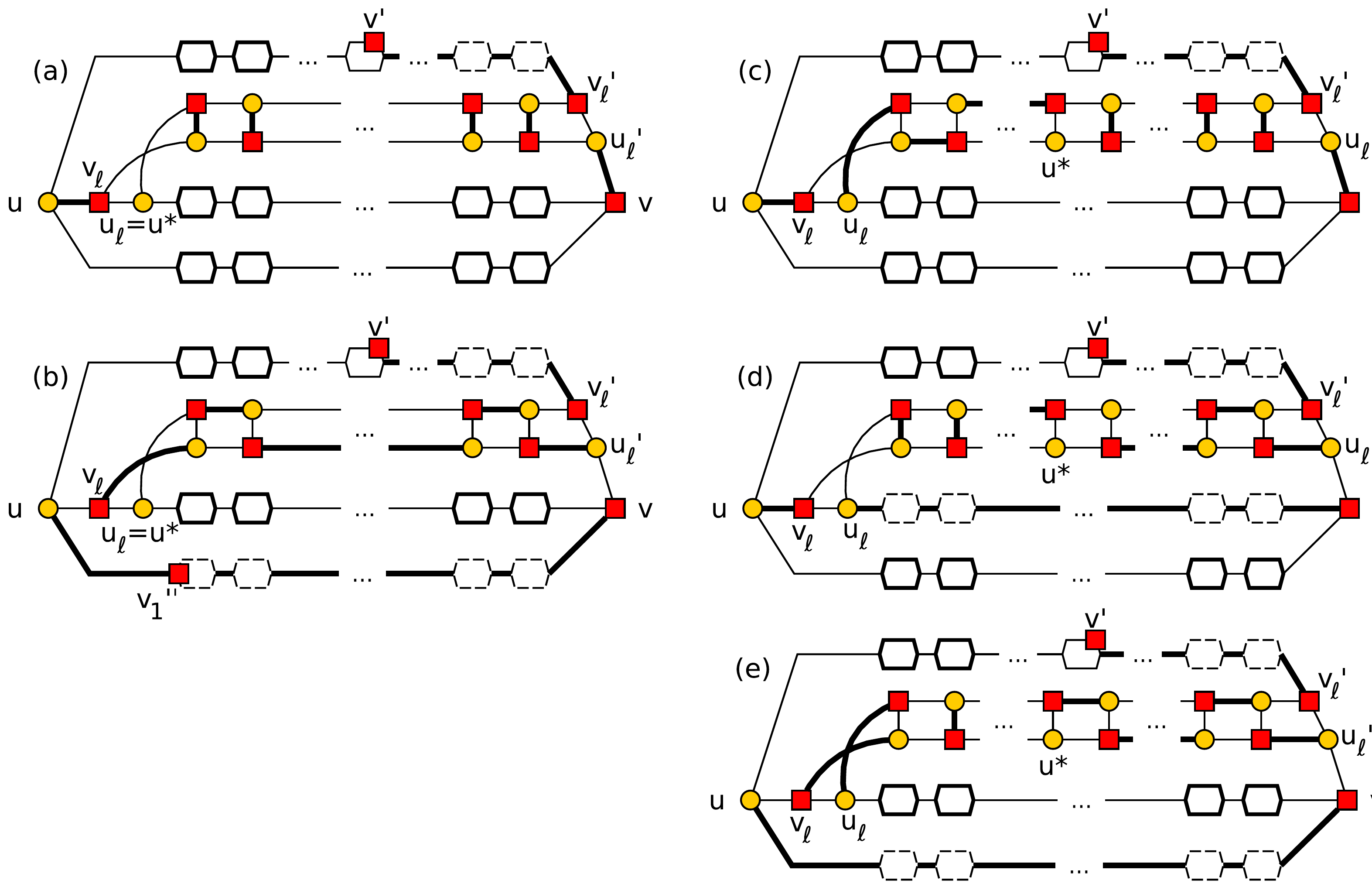}
\caption{Schematic figure of case distinctions in case~5}
\label{figure:CaseFiveLadder}
\end{figure}
\end{description}
If we put all five cases together we find with inequality (6),
$$|N_{u^*v'}|\leq 3\cdot 6^{2k+k'}3^{k'}F_{k+3}+2  \cdot 6^{3k}\leq 3 \cdot 2^4 \cdot 6^{3k}2^{\frac{k}{2}}+2 \cdot 6^{3k} \leq 50 \cdot 6^{3k} \cdot 2^{k/2}.$$
We get the upper bound~$|S|\leq 400 n 6^{3k}2^{\frac{k}{2}}$ if we take into account that there exist~$4n$ sets~$N_{u^*v'}$ and all~$4n$ sets~$N_{u'v^*}$ can be handled analogously. If we consider a set~$\mathcal{P}'$ of simple paths between all vertex pairs~$(M,M')$ of~$N_{uv}^1\times N_{uv}^2$ (exactly one between each pair), then we find that~$|\mathcal{P}'|=6^{6k-2}(F_{k+3})^2\cdot 9$ paths use an arc in~$\Gamma$ which is adjacent to the vertex set~$S$. Note that each vertex in~$S$ is adjacent to at most~$m$ arcs (loops are possible). Hence, there exists an arc~$a$ in~$\Gamma$ which has to be used from at least~$\frac{9 \cdot 6^{6k-2}F_{k+3}^2}{400 \cdot m \cdot n 6^{3k}2^{\frac{k}{2}}}=\frac{9 \cdot 6^{3k-2}F_{k+3}^2}{400mn \cdot 2^{\frac{k}{2}}}$ paths of~$\mathcal{P}'$.
With equation~(\ref{eqn:ContributionMaximumLoading}) we get for the maximum loading for an arbitrary set $\mathcal{P}$ of paths in $\Gamma,$
\footnotesize
$$\rho_1(\mathcal{P})\geq \rho_1(\mathcal{P}')\geq \frac{9 \cdot 6^{3k-2}F_{k+3}^2}{400 m n \cdot 2^{\frac{k}{2}}} \cdot \frac{m}{|\Omega|} \geq \frac{9 \cdot 6^{3k-2}F_{k+3}^2}{400n \cdot 2^{\frac{k}{2}}|N_{uv}|(n^2+1)}=\frac{F_{k+3}}{1600 n(n^2+1)\cdot 2^{\frac{k}{2}}}.$$
\normalsize
With inequality (\ref{eqn:Fibonacci}) we find
$$\frac{F_{k+3}}{2^{\frac{k}{2}}} \geq \frac{(1+\sqrt{5})^{k+3}}{2^{k+3}\sqrt{5}\sqrt{2^k}}=\frac{(1+\sqrt{5})^3}{2^3\sqrt{5}}\left(\frac{1+\sqrt{5}}{2\sqrt{2}}\right)^k.$$
Since $\frac{1+\sqrt{5}}{2 \sqrt{2}}>1$ we find that the maximum loading cannot be lower bounded by a polynomial. Together with Proposition~\ref{Proposition:lowerBoundMixingtime} we find that the graph class does not possess a polynomial mixing time.\qed
\end{proof}

We point out that planar graphs and threshold graphs belong to `difficult graph classes' when using Broder's chain in contrast to the existence of efficient exact samplers for these cases. Until now, we did not find a bipartite graph class where the fraction~$|N(G)|/|M(G)|$ is not  bounded polynomially for each graph~$G$ while the mixing time of Broder's chain is. We conjecture that the missing logical direction of the result in \cite[Corollary 3.13]{Sinclair:1993:ARG:140552} in Sinclair’s PhD-thesis can be formulated as follows.

\begin{conjecture}
For a bipartite graph class Broder's chain is rapidly mixing if and only if the fraction $|N(G)|/|M(G)|$ is bounded polynomially for each~$G$.
\end{conjecture}

\section{Experimental Insights}
\label{sec:experiments}

The best known upper bounds for mixing times shown in Table~\ref{table:mixingtimes} are far too large to be used in practice. We want to investigate the quality of these bounds to see if they fit the exact total mixing time well or if there is some potential for further improvement. The upper bounds in Table~\ref{table:mixingtimes} are general, that means they make only little use of the state graph's very special structure. We believe that knowing the structure of a state graph is essential for getting tight bounds. 
To judge the quality of the upper bounds we firstly compute the total mixing time of a Markov chain. Of course this quantity heavily depends on the size and structure of the given bipartite graph and is not known in general, so we have to compute the total mixing time for each bipartite graph individually. We also apply the methods for bounding the mixing time to each graph individually to see if considering the structure of a graph may improve the quality of a bound. 
For our experiments, we consider the spectral bound (see inequality~(\ref{eqn:spectral_bound})) as well as the multicommodity bound by flow function $f_2$ (see inequality~(\ref{eqn:flow2})). The multicommodity bound depends on a set of paths~$\mathcal{P}$. We investigate three different path sets.
\begin{description}
\item[Sinclair's canonical paths \boldmath$\mathcal{P}_1$:] Between each pair of matchings $M, M' \in \Omega$ exactly one path is constructed. Each path consists of three segments. An initial segment connects $M$ with the nearest perfect matching $\bar{M} \in M(G)$. A main segment connects $\bar{M}$ with a perfect matching $\bar{M}' \in M(G)$ with minimal distance to $M'$. The end segment connects $\bar{M}'$ to $M'$. For details see~\cite{Jerrum:1989:AP:76071.76077}.
\item[One shortest path \boldmath$\mathcal{P}_2$:] We send the flow $f_2(M,M')$ on exactly one shortest path between~$M$ and~$M'$. 
\item[All shortest paths \boldmath$\mathcal{P}_3$:] We divide the flow $f_2(M,M')$ uniformly between \emph{all} shortest paths between $M$ and $M'$.
\end{description}
For our experiments, we always set $\epsilon = 1 / (2e)$ as proposed by Sinclair \cite{SinclairLectures}. Note that this is a reasonable choice in many practical applications. In a hypothetical state graph with one thousand states a total variation distance of $\epsilon = 1 / (2e)$ implies a sampling probability which varies from stationary distribution by approximately $0.00037$ on average per state. 

\subsection{Experiments and Results}

\paragraph{First Experiment}
We enumerate all connected, non-isomorphic bipartite graphs with $2n=12$ vertices. From all the $212,780$ bipartite graphs, $89,242$ graphs contain at least one perfect matching. For each graph the corresponding state graph has been constructed. (Note that the number of bipartite graphs grows exponentially with $n$. With $2n=14$ this number exceeds 13 million, which is clearly too much for a complete experimental analysis.) Figure~\ref{fig:distr} shows the cumulative distribution of the state graph size~$|\Omega|$.

\begin{figure}[t]
  \begin{center}
    \includegraphics[width=0.5\textwidth]{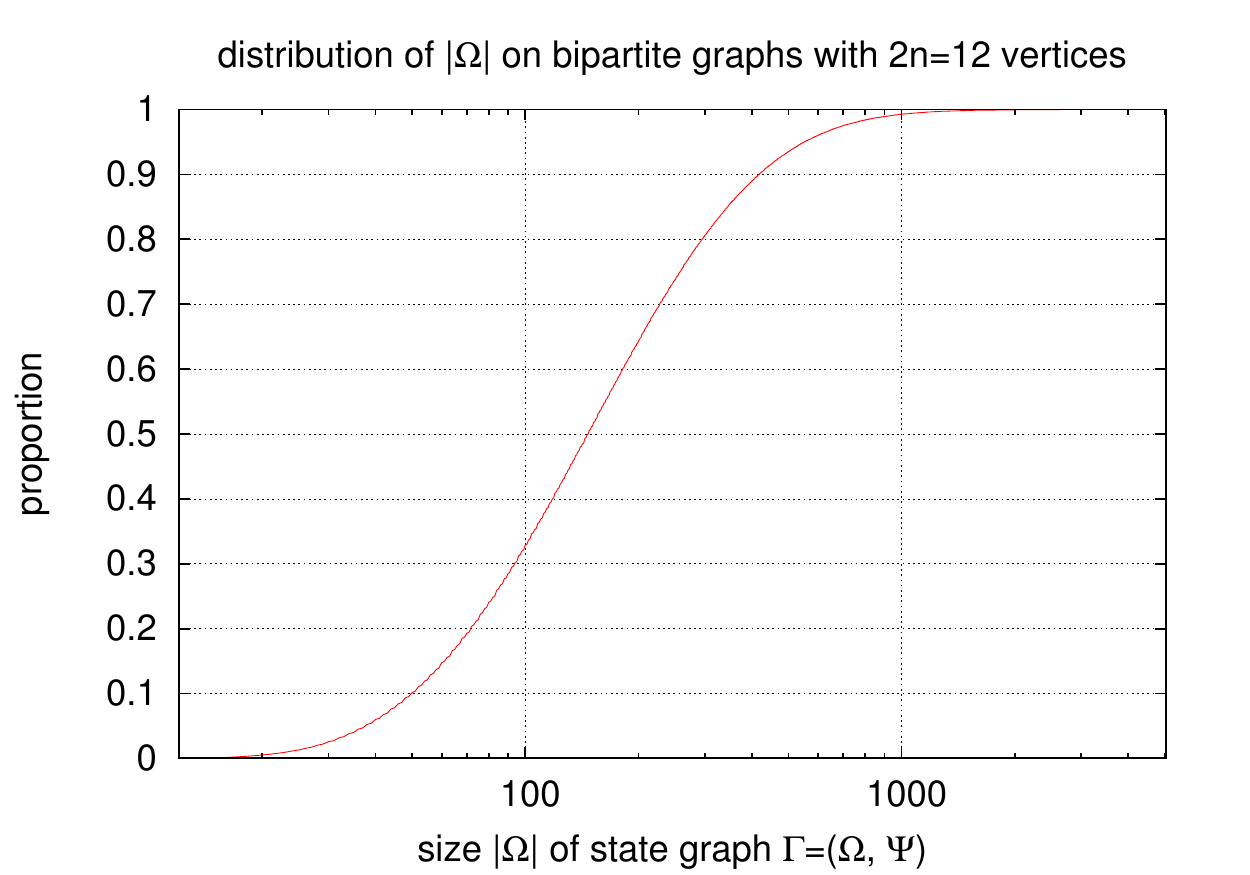}
    \caption{Cumulative distribution of state graph size~$|\Omega|$ over all connected bipartite graphs with 12 vertices and at least one perfect matching. Minimum:~12, maximum:~5040, average:~203.1, median:~148}
    \label{fig:distr}
  \end{center}
\end{figure}

Figure~\ref{fig:size_mixing} shows the size of a state graph versus the corresponding mixing time. These two dimension do not seem to correlate. A first surprising observation is that in case of the JSV-chain the largest state graph has the smallest total mixing time (see Figure~\ref{fig:size_mixing}). In both Markov chains, the maximum mixing time is taken at relatively small state graphs. This observation suggests that the structure of a state graph has a greater influence on the mixing time than its size.

\begin{figure}
	\centering
	\begin{subfigure}[b]{0.49\textwidth}
		\includegraphics[width=\textwidth]{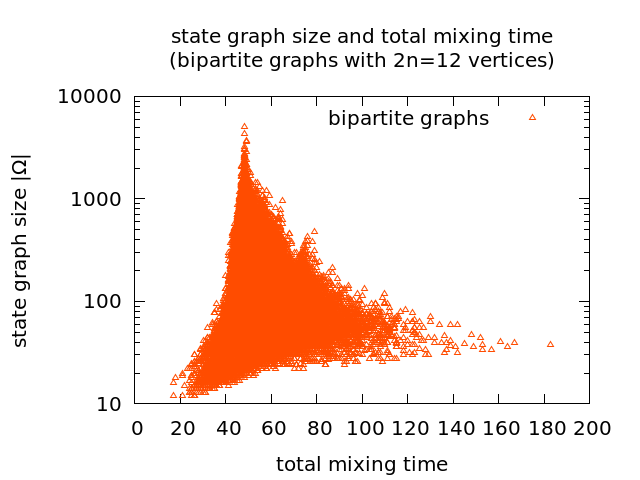}
		\caption{Broder's chain}
	\end{subfigure}
	\begin{subfigure}[b]{0.49\textwidth}
		\includegraphics[width=\textwidth]{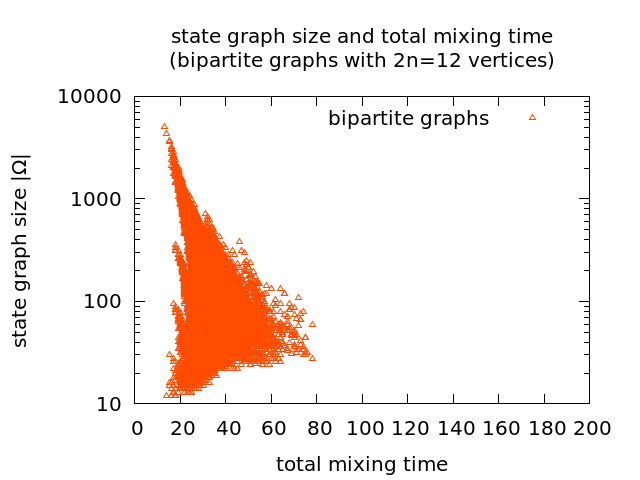}
		\caption{JSV-chain}
	\end{subfigure}
	\caption{State graph size versus mixing time for all connected bipartite graphs with $12$ vertices and at least one perfect matching}
	\label{fig:size_mixing}
\end{figure}

For each graph we compute the exact total mixing time $\tau(\epsilon)$, the spectral bound as well as three  multicommodity bounds using flow $f_2$ and the set of paths $\mathcal{P}_{1}$, $\mathcal{P}_{2}$ and $\mathcal{P}_{3}$. We investigate how far these bounds lie apart.

Comparing the total mixing time with its corresponding upper bounds, we gain the result shown in Figure~\ref{fig:experiment1}. We observe that the total mixing time is far smaller than the multicommodity bounds. That seems to be enough evidence for us to believe that theoretic upper bounds in Table~\ref{table:mixingtimes} are too pessimistic. Moreover, very different graph classes realize worst case total mixing times for the four bounding methods. We observe a clear ranking.
\begin{description}
\item[1. The spectral bound] is the most accurate one.
\item[2. The multicommodity bound \boldmath$\mathcal{P}_3$] is the most promising one after the spectral bound.
\item[3. The multicommodity bound \boldmath$\mathcal{P}_2$] cannot be better than the multicommodity bound $\mathcal{P}_3$, because the same flow is concentrated on a single path.
\item[4. The multicommodity bound \boldmath$\mathcal{P}_1$] has the poorest quality. Intuitively this can be explained by the fact that the paths of the set~$\mathcal{P}_1$ are at least as long as the paths in~$\mathcal{P}_2$. So, the same amount of flow is carried over more arcs.
\end{description}

\begin{figure}
	\centering
	\begin{subfigure}[b]{0.49\textwidth}
		\includegraphics[width=\textwidth]{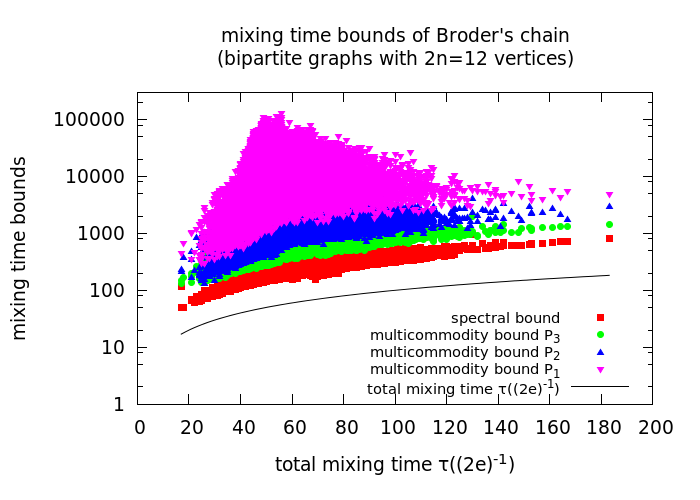}
	\end{subfigure}
	\begin{subfigure}[b]{0.49\textwidth}
		\includegraphics[width=\textwidth]{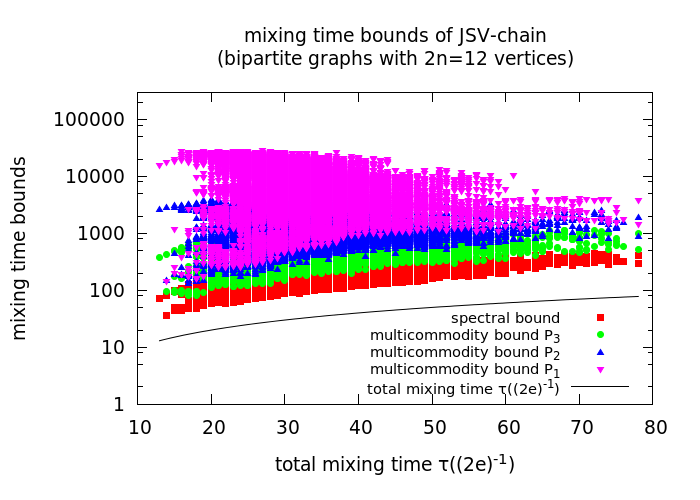}
	\end{subfigure}
	\\
	\centering
	\begin{subfigure}[b]{0.49\textwidth}
		\includegraphics[width=\textwidth]{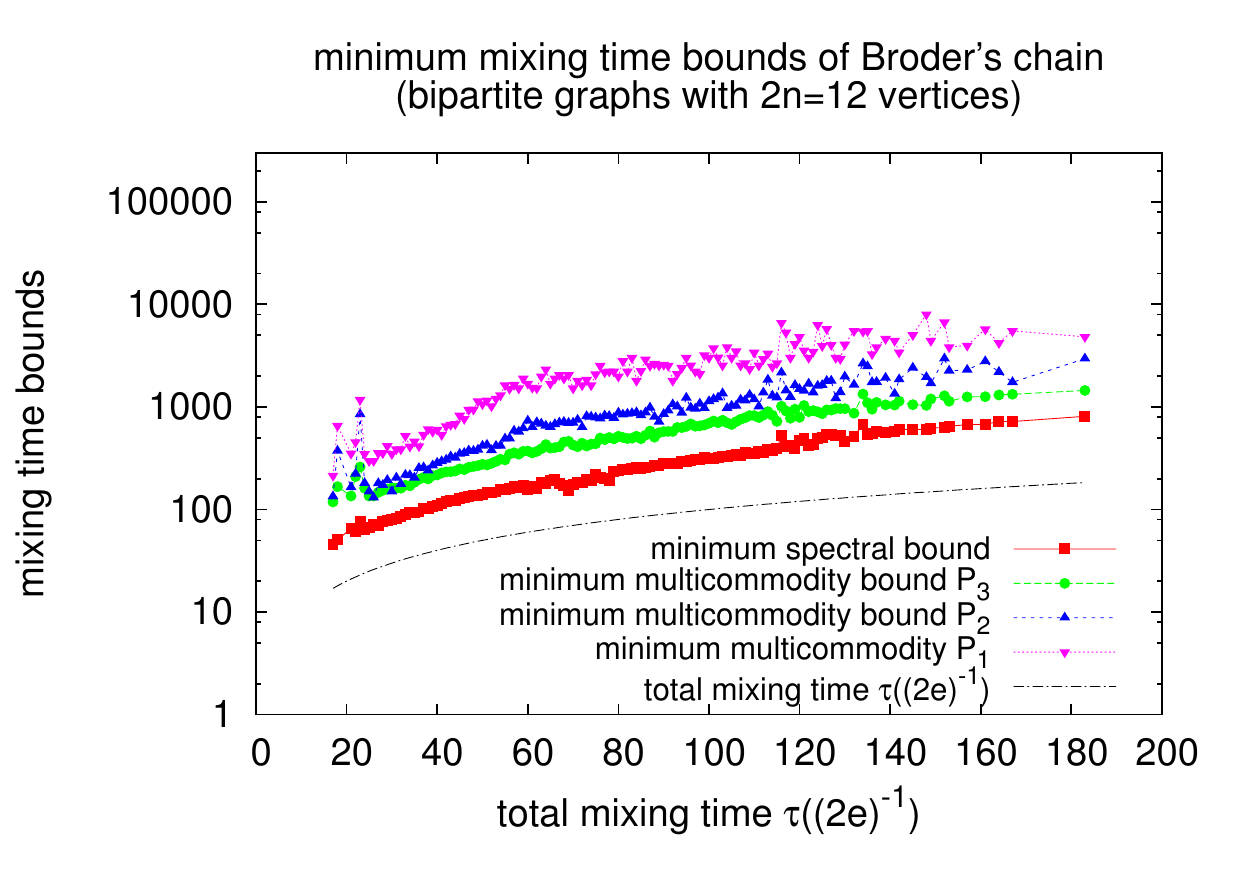}
	\end{subfigure}
	~
	\begin{subfigure}[b]{0.49\textwidth}
		\includegraphics[width=\textwidth]{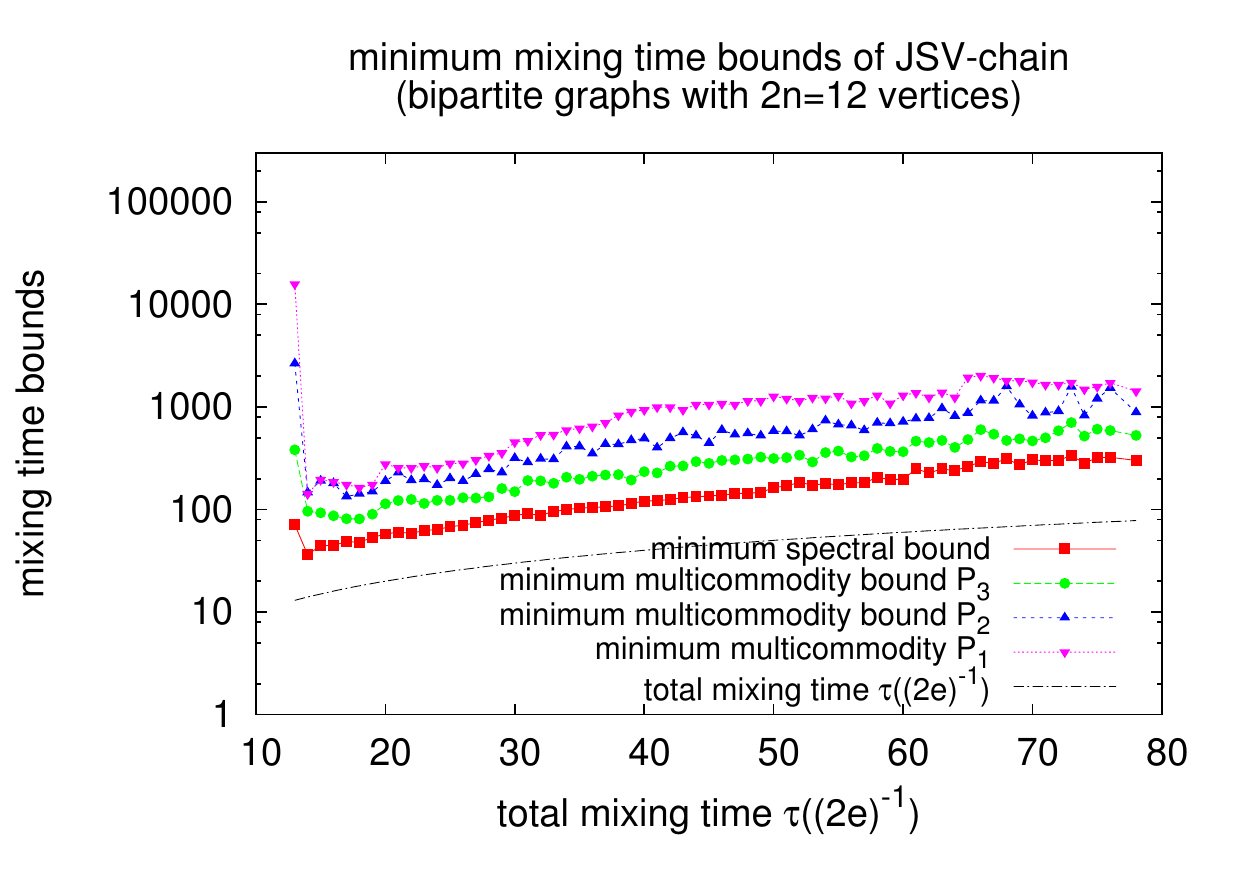}
	\end{subfigure}
	~
	\begin{subfigure}[b]{0.49\textwidth}
		\includegraphics[width=\textwidth]{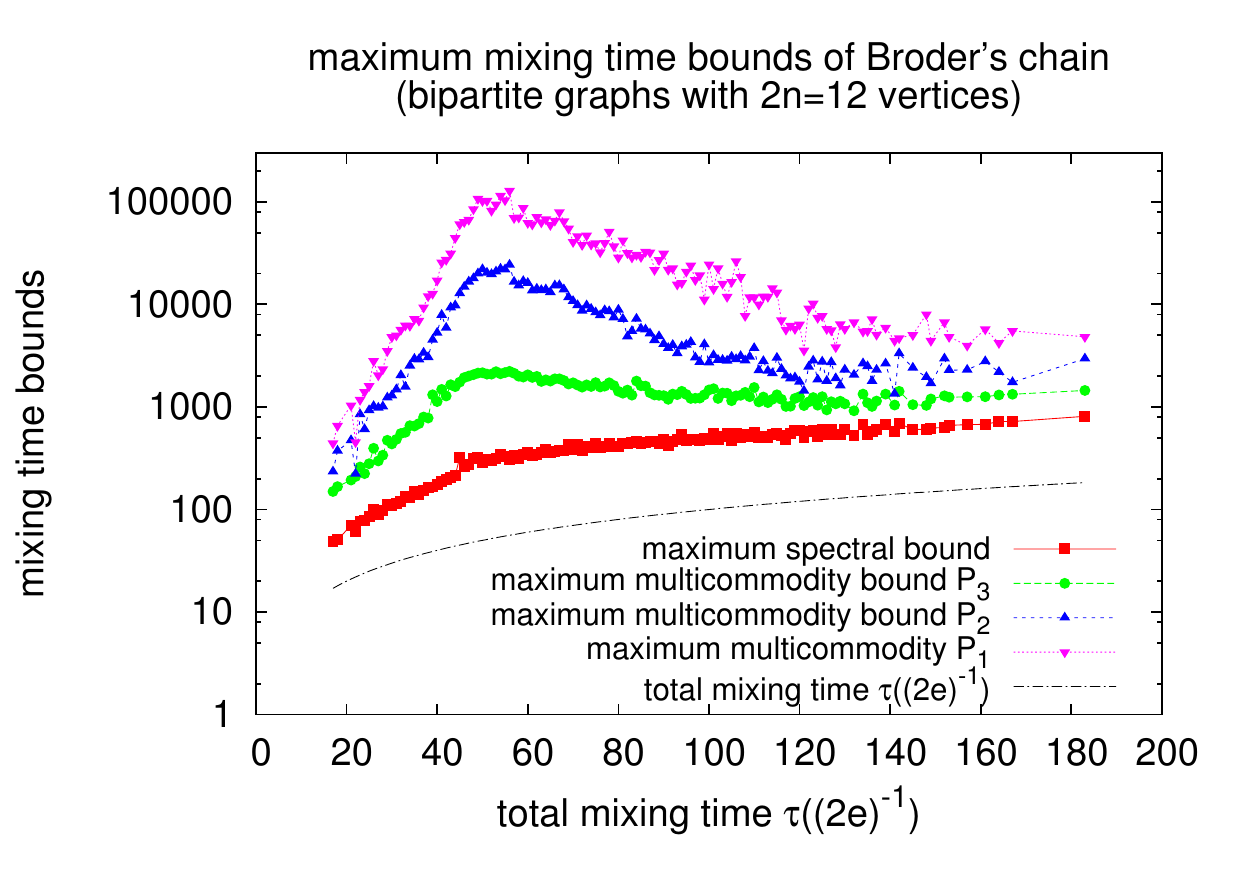}
	\end{subfigure}
	~
	\begin{subfigure}[b]{0.49\textwidth}
		\includegraphics[width=\textwidth]{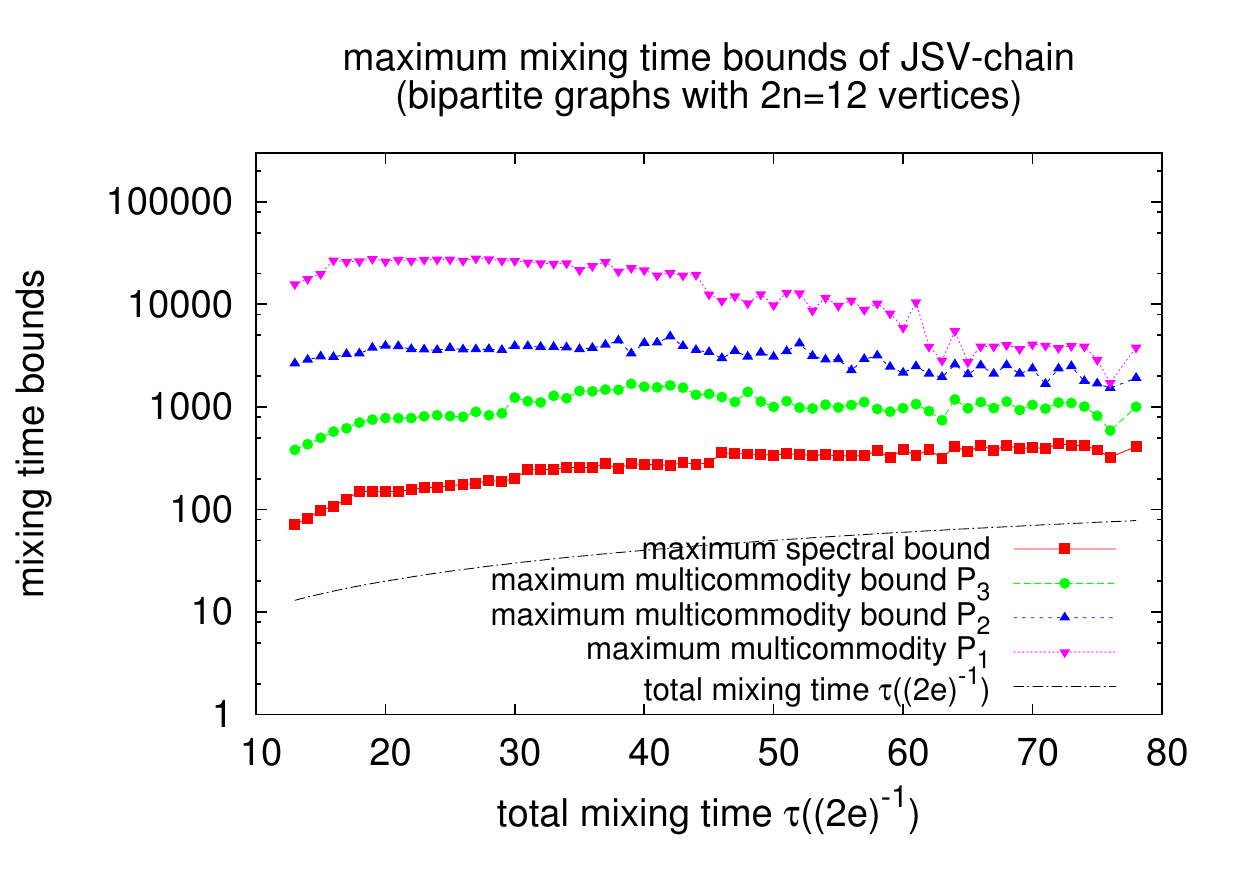}
	\end{subfigure}
	~
	\begin{subfigure}[b]{0.49\textwidth}
		\includegraphics[width=\textwidth]{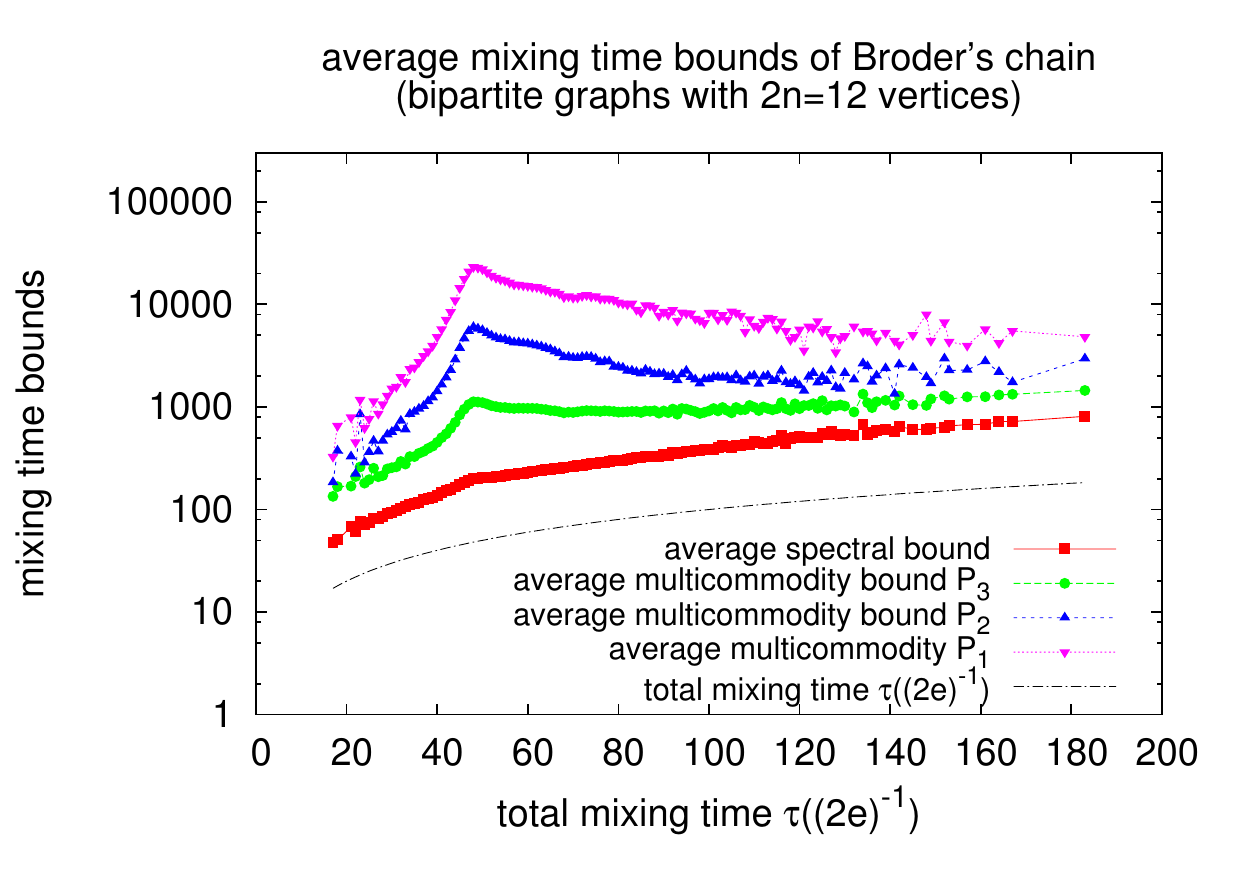}
		\caption{Broder's chain}
	\end{subfigure}
	~
	\begin{subfigure}[b]{0.49\textwidth}
		\includegraphics[width=\textwidth]{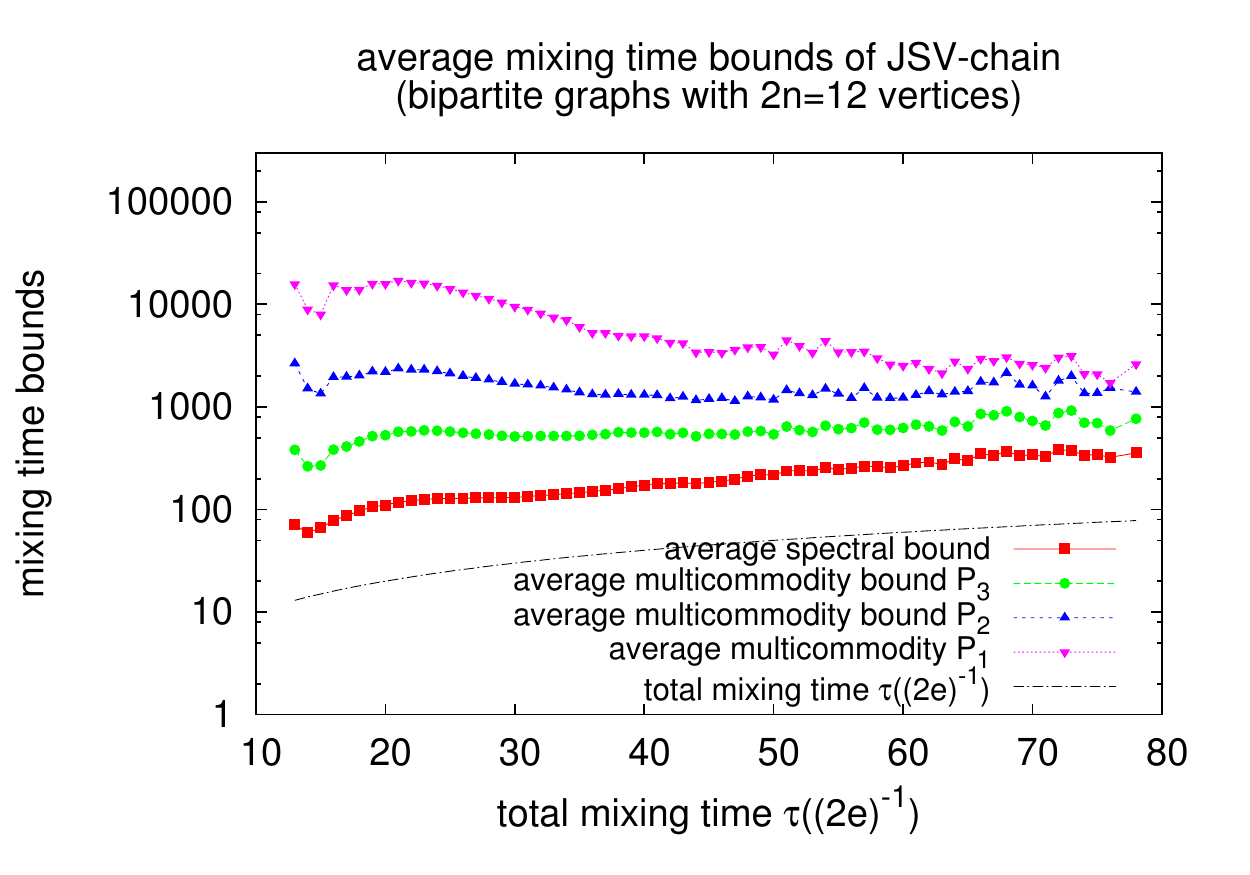}
		\caption{JSV-chain}
	\end{subfigure}
	\caption{Upper bounds on mixing time and minimum, maximum and average bounds compared to the total mixing time for all connected bipartite graphs with~$12$ vertices and at least one perfect matching}
	\label{fig:experiment1}
\end{figure}

\paragraph{Second Experiment} In a second experiment we want to study the difference between Broder's chain and the JSV chain and the quality of the upper bounds when the number of vertices $n$ grows. We consider the class of \emph{hexagon graphs} (see Figure~\ref{figure:hexagon}) for a different number $k$ of hexagons. For each hexagon graph we construct the corresponding state graph $\Gamma$ and compute the bounds as in the first experiment. With our system and implementation, we were able to compute the total mixing time and its bounds for the hexagon graph for $k=1, \ldots, 8$, investing approximately three days of computing time. Most of this time is spent on matrix multiplications during the computation of the exact total mixing time and for the computation of the multicommodity bound $\mathcal{P}_3$. 
Figure \ref{fig:experiment2} shows the results. Note the logarithmic y-axes and the growing gap between total mixing time and the upper bounds. As theory predicts, the mixing time of Broder's chain (not rapidly mixing, see Proposition~\ref{Proposition:hexagonGraph}) grows a lot faster than its equivalent in the JSV-chain (rapidly mixing \cite{JerrumSinclairVigoda04}). The gap between each of the upper bounds strongly increases with growing $k$ in both Markov chains. The ranking of the bounds from the first experiment is confirmed. 

Additional experiments with smaller $\epsilon$ (see Figures~\ref{fig:experiment1_eps1E9} and \ref{fig:experiment2_eps1E9} in the appendix) confirm the ranking of the bounds (which follows immediately from theory) but also show that the exact total mixing time comes closer to the upper bounds.

\begin{figure}
	\centering
	\begin{subfigure}[b]{0.49\textwidth}
		\includegraphics[width=\textwidth]{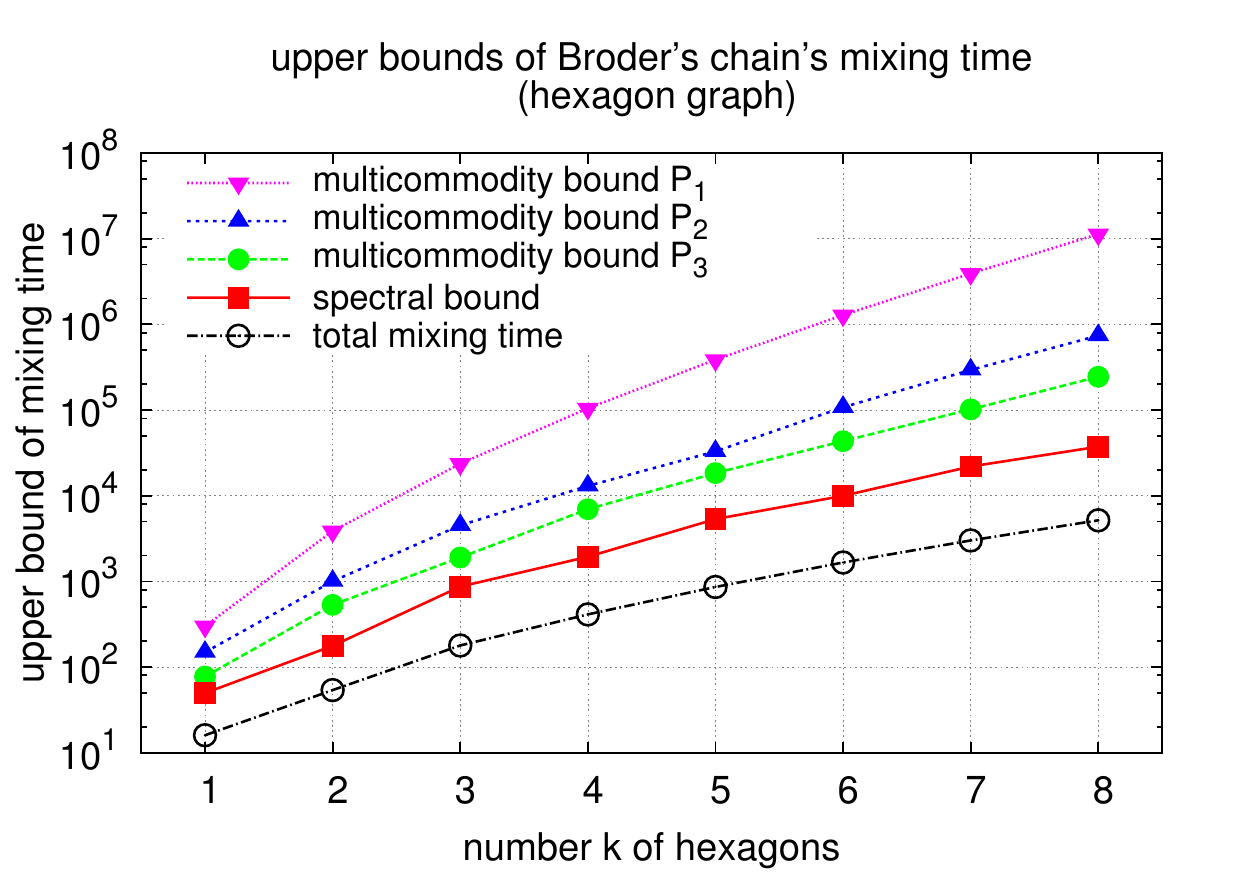}
		\caption{Broder's chain}
	\end{subfigure}
	\begin{subfigure}[b]{0.49\textwidth}
		\includegraphics[width=\textwidth]{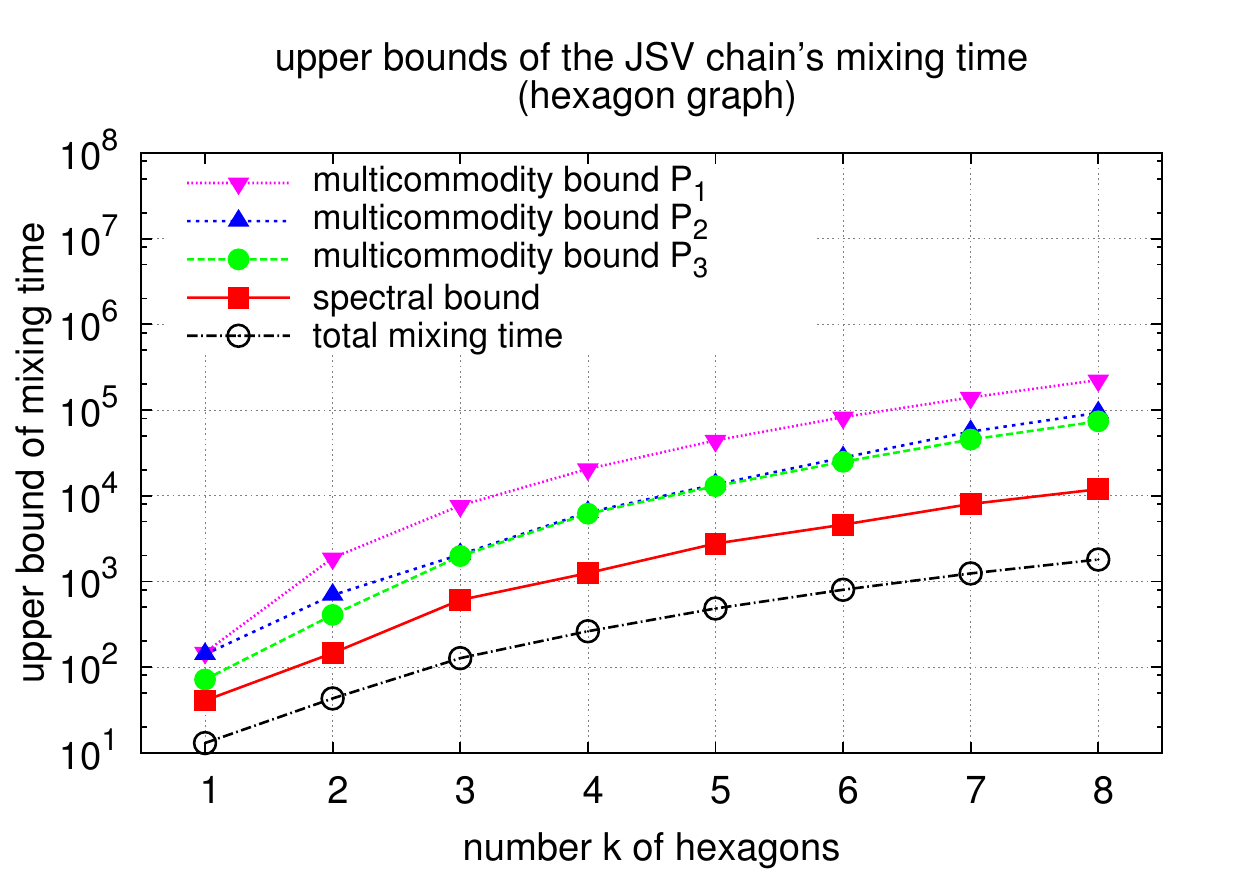}
		\caption{JSV's chain}
	\end{subfigure}
	\caption{Upper bounds of mixing time for hexagon graphs with small~$k$}
	\label{fig:experiment2}
\end{figure}

\subsection{Implementation Details}

\paragraph{Constructing the state graph}
We iteratively construct the state graph $\Gamma = (\Omega, \Psi)$ of a bipartite graph $G$ by running a full graph scan, starting with an arbitrary $M \in M(G)$. At each step of the construction phase, we consider a  non-visited matching $M$ and construct its set of out-neighbors $N^+(M) = \{ M' \colon P(M,M') > 0 \}$. We mark $M$ as visited and continue the construction by setting $M \gets M'$ for each neighbor $M' \in N^+(M)$. We stop, if every state is visited. During the process we count $|M(G)|$ and $|N_{u,v}(G)|$ for $(u,v) \in U \times V$. 
\paragraph{Graph enumeration}
We use the command line tool \emph{nauty} \cite{McKay201494} to enumerate all non-isomorphic connected bipartite graphs with $12$ vertices. 
\paragraph{Total mixing time}
At a fixed $t$, total variation distance $d(\pi, t)$ can be computed by considering $P^t$, the $t$h power of transition matrix $P$. Each row $(P^t(x_i,y))_{y \in \Omega}$ refers to $p_{x_{i}}^{t}$, the probability distribution at time $t$ with start distribution $p_{x_{i}}^{0}$, where the $i$-th component is one and all other components are zero. We compute the variation distances $d(\pi,p_{x_{i}}^{t})$ for all $i \in \{1,\dots,|\Omega|\}$ and~$d(\pi,t).$  Note that $d(\pi,t)$ is monotonically decreasing with $t$ \cite{SinclairLectures}, so we can use binary search to find $\tau(\epsilon)$, requiring only a logarithmic number of matrix multiplications. 
\paragraph{Computing the spectral bound}
The $(|\Omega| \times |\Omega|)$-transition matrix $P$ has at most $m|\Omega|$ non-zero entries. We use the sparse matrix functionality of the \emph{Python} package \emph{scipy} to compute the (in absolute terms) largest three eigenvalues.
\paragraph{Computing a multicommodity bound}
We compute the set of paths $\mathcal{P}_1$, $\mathcal{P}_2$, $\mathcal{P}_3$ for each state graph~$\Gamma$. Each path~$\mathcal{P}_{MM'} \in \mathcal{P}_i$, containing arc $a$, contributes an amount of $(f_2(\mathcal{P}_{MM'})/Q(a))$ to the loading $(f_2(a)/Q(a))$ of each arc $a \in \Psi$ (see equation~(\ref{eqn:maximumLoading})). We sum over all paths in $\mathcal{P}_i$.
\paragraph{Computing the congestion of path sets}
The computation of the congestion by $\mathcal{P}_1$ and $\mathcal{P}_2$ can be done by explicitly constructing exactly one path between each pair of states.
In contrast, the set~$\mathcal{P}_3$ contains all shortest paths between each state pair. Explicitly constructing this set requires an exponential amount of memory which we avoid by using the following approach.

For each state $x$ we compute a shortest-path-DAG rooted at~$x$ by depth first search. We use this DAG to compute the number of shortest paths $|\mathcal{P}_{x,y}|$ from $x$ to each state $y$. For each~$y\in\Omega$, the number $|\mathcal{P}_{x,y}|$ can be counted efficiently as the sum of $|\mathcal{P}_{x,v}|$ where $(v,y)$ is an arc of the shortest-path-DAG. By traversing the DAG in topological order we compute $|\mathcal{P}_{x,y}|$ for all $y\in\Omega$ in linear time. For each~$y$ we also compute the number of shortest paths $|\mathcal{P}_{v,y}|$ from each state~$v\in\Omega$ to~$y$. This numbers  can be computed analogously by traversing the DAG in counter topological order. 
In the next step we add an amount of $\delta := f_2(x,y) / |\mathcal{P}_{x,y}|$ to the loading of each arc $a=(u,v)$ for each shortest path in $\mathcal{P}_{x,y}$ which contains $a$. In other words, the loading of arc $a$ is increased by $m_a \cdot \delta$ where~$m_a$ is the number of paths in $\mathcal{P}_{x,y}$ using $a$. This number can be computed as the product of~$|\mathcal{P}_{x,u}|$ and~$|\mathcal{P}_{v,y}|$. By iterating over all~$x$ and~$y$ and traversing the DAG for each fixed pair of $x$ and $y$, we obtain an $O(|\Omega|^3)$ algorithm for computing the congestion by $\mathcal{P}_3$.
\\

Unless stated otherwise, the algorithms are implemented in \emph{Java}, version 6. For matrix multiplication we use the package \emph{JBLAS}, version 1.2. The experiments were conducted on a standard \emph{Ubuntu} 12.04 system with a \emph{Intel(R) Xeon(R) CPU X5570 $@$ $2.93$GHz} processor and~$48$ gigabyte of memory.

\section{Conclusion}

We proved that Broder's chain is not rapidly mixing for several graph classes. Multicommodity bounds are too weak to get practicable mixing times not only when knowledge about the structure of state graphs is missing but also when the structure is known explicitly. On the other hand the spectral bound doesn't differ so much from the exact total mixing time in our experiments. For future work we recommend to investigate the structure of state graphs and to develop suitable results in spectral graph theory.

\bibliography{references}

\newpage
\begin{appendix}
\section*{Appendix}

The proof of the~$\#P$-completeness of the computation of~$\frac{N(G)}{M(G)}$ can be done in three steps.

\begin{proposition}
\label{prop1}
Let~$G=(U \cup V, E)$ be a bipartite graph with~$|M(G)| > 0$. The computation of the fraction~$\frac{|N_{u,v}(G)|}{|M(G)|}$ is \emph{\#P-complete} for arbitrary~$u \in U$ and~$v \in V$.
\end{proposition}

\begin{proof}
We assume, we could efficiently compute the ratio~$|N_{u,v}(G)|/|M(G)|$ for a given bipartite graph~$G$. Let~$M = \{ e_1, \ldots, e_n \}$ be an arbitrary perfect matching in~$M(G)$ with~$e_i = \{u_i, v_i\}$. We show how to compute~$|M(G)|$ from a polynomial number of ratios~$|N_{u,v}(G)|/|M(G)|$ of suitable bipartite graphs. To this end, we introduce a set of auxiliary graphs~$G_0, \ldots, G_{n-1}$ with~$G_i~=~(U_i~\cup~V_i,~E_i)$ for~$ 0 \leq i < n$. The graph~$G_i$ is constructed from~$G_{i-1}$ by the following recurrence
\begin{eqnarray*}
G_0 &=& G \\
G_i &=& \left(U_{i-1} \cup V_{i-1} \setminus \{ u_i, v_i \}, E_{i-1} \setminus \big\{ \{ x, y \}~|~ x = u_i \text{ or } y = v_i \} \big\} \right).
\end{eqnarray*}

Note that~$M_i = M \setminus \{ e_1, \ldots, e_i \}$ is a perfect matching in~$G_i$ and so~$M(G_i) \not= 0$ for all~$0 \leq i < n$.
Furthermore, for all~$0 < i < n$ we find
\[ |M(G_i)| = |N_{u_i,v_i}(G_{i-1})|, \]
because the number of near-perfect matchings in~$G_{i-1}$ where~$u_i$ and~$v_i$ are the only unmatched vertices equals the number of perfect matchings in~$G_i$. There is exactly one near-perfect matching in the graph with zero edges, i.e.~$|N_{u_n,v_n}(G_{n-1})| = 1$.
Consider the telescope product

\begin{eqnarray*}
|M(G)| &=& \frac{|M(G_0)|}{|N_{u_n,v_n}(G_{n-1})|} \\
&=& \frac{|M(G_0)|}{|N_{u_1,v_1}(G_0)|} \cdot \frac{|N_{u_1,v_1}(G_0)|}{|N_{u_2,v_2}(G_1)|} \cdot \ldots \cdot \frac{|N_{u_{n-1},v_{n-1}}(G_{n-2})|}{|N_{u_n,v_n}(G_{n-1})|} \\
&=& \frac{|M(G_0)|}{|N_{u_1,v_1}(G_0)|} \cdot \frac{|M(G_1)|}{|N_{u_2,v_2}(G_1)|} \cdot \ldots \cdot \frac{|M(G_{n-1})|}{|N_{u_n,v_n}(G_{n-1})|}.
\end{eqnarray*}

Assume we could efficiently compute the ratio~$|M(G_i)| / |N_{u_{i+1},v_{i+1}}(G_i)|$ for each~$G_i$. This leads to an efficient way to compute~$|M(G)|$, which contradicts the \#P-completeness of this problem \cite{Valiant79}.
\qed
\end{proof}

\begin{proposition}
\label{prop2}
Let~$G=(U \cup V,E)$ be a bipartite graph with~$|M(G)|>0$. Let~$u^* \in U$ be an arbitrary vertex. The computation of~$\sum\limits_{v \in V}\frac{|N_{u^*,v}(G)|}{|M(G)|}$ is \emph{\#P-complete}.
\end{proposition}

\begin{proof}
We show that an efficient computation of~$\sum_{v \in V} |N_{u^*,v}(G)| / |M(G)|$ would lead to an efficient computation of~$|N_{u^*,v^*}(G)| / |M(G)|$ for each vertex~$v^* \in V$, in contradiction to Proposition \ref{prop1}. To do so, we fix an arbitrary vertex~$v^* \in V$ and append the paths~$P_{u^*} = (u^*,v',u'')$ and~$P_{v^*} = (v^*, u', v'')$ to~$u^*$ respectively~$v^*$, receiving the graph~$G' = (U' \cup V', E')$ with~$ U' = U \cup \{ u', u''\}$,~$V' = V \cup \{ v', v'' \}$ and~$E' = E \cup \big\{ \{ u^*, v' \}, \{ u', v'' \}, \{ v^*, u' \}, \{ v', u'' \} \big\}$ (see Figure \ref{figure:graphSharpP2}).
\begin{figure}[tbp]
\centering
\includegraphics[width=10cm]{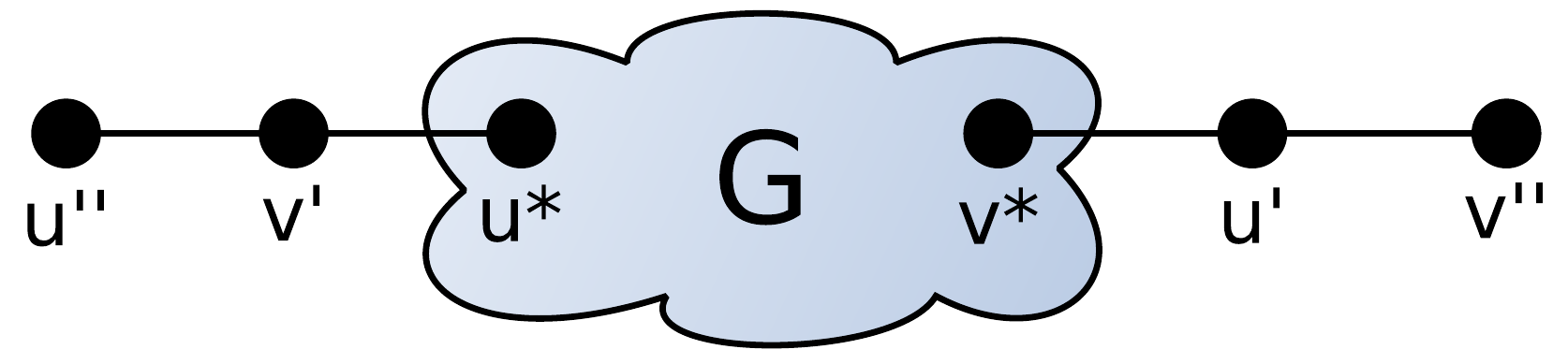}
\caption{Schematic picture for the construction of graph~$G'$}
\label{figure:graphSharpP2}
\end{figure}
The number of perfect matchings in~$G'$ equals the number of perfect matchings in~$G$, so~$|M(G')| = |M(G)|$.
We now consider~$ \sum_{v \in V'} |N_{u'',v}(G')|~$, the sum of all the near-perfect matchings in~$G'$ where the designated vertex~$u''$ remains unmatched. We find 
\[ \sum_{v \in V'} |N_{u'',v}(G')| = \sum_{v \in V} |N_{u'',v}(G')| + |N_{u'',v'}(G')| + |N_{u'',v''}(G)'|. \]
Each of the summands~$|N_{u'',v}(G')|$ for~$v \in V$ equals~$|N_{u^*,v}(G)|$, because~$u''$ being unmatched enforces~$v'$ to be matched with~$u^*$ to get a near-perfect matching in~$G'$. The second summand~$|N_{u'',v'}(G')|$ equals~$|M(G)|$ because every perfect matching in~$G$ plus the edge~$ \{ v', u'' \}$ makes a near-perfect matching in~$G'$ leaving~$u''$ and~$v'$ unmatched. The third summand~$|N_{u'',v''}(G')|$ is equal to~$|N_{u^*,v^*}(G)|$, because~$u''$ and~$v''$ being unmatched implies that~$u'$ and~$v'$ are matched with~$u^*$ and~$v^*$, respectively. We get
\[ \sum_{v \in V'} |N_{u'',v}(G')| = \sum_{v \in V} |N_{u^*,v}(G)| + |M(G)| + |N_{u^*,v^*}(G)|. \]
Dividing each side by~$|M(G)|$ leads to
\[ \sum_{v \in V'} \frac{|N_{u'',v}(G')|}{|M(G')|} = \sum_{v \in V} \frac{|N_{u^*,v}(G)|}{|M(G)|} + 1 + \frac{|N_{u^*,v^*}(G)|}{|M(G)|}. \]
An efficient way of calculating the ratios~$\sum_{v \in V'}\frac{ |N_{u'',v}(G')|}{|M(G')|}$ and~$\sum_{v \in V} \frac{ |N_{u^*,v}(G)|}{|M(G)|}$ leads to an efficient way of calculating~$\frac{|N_{u^*,v^*}(G)|}{|M(G)|}$.
\qed
\end{proof}

\begin{theorem}\label{SharpComplete}
Let~$G=(U \cup V,E)$ be a bipartite graph with~$|M(G)| > 0$. The computation of~$\frac{|N(G)|}{|M(G)|}$ is \emph{\#P-complete}.
\end{theorem}

\begin{proof}
We construct the auxiliary graph~$\tilde{G} = (\tilde{U} \cup \tilde{V}, \tilde{E})$ by adding a path~$P_{v_i} = (v_i, u_i', v_i'')$ to every vertex~$v_i \in V$ and a path~$P_{u_i} = (u_i, v_i', u_i'')$ to every vertex~$u_i \in U \setminus \{ u^* \}$ for arbitrary~$u^*$. For a schematic example consider Figure~\ref{figure:graphSharpP}. Formally, we get the vertex sets~$\tilde{U} = U \cup U' \cup U''$ and~$\tilde{V} = V \cup V' \cup V''$ where
\begin{eqnarray*}
U' &=& \big\{ u' : v \in V \big\} \\
V' &=& \big\{ v' : u \in U \setminus \{ u^* \} \big\} \\
U'' &=& \big\{ u'' : v' \in V' \big\} \\
V'' &=& \big\{ v'' : u' \in U' \big\}.
\end{eqnarray*}

\begin{figure}[tbp]
\centering
\includegraphics[width=10cm]{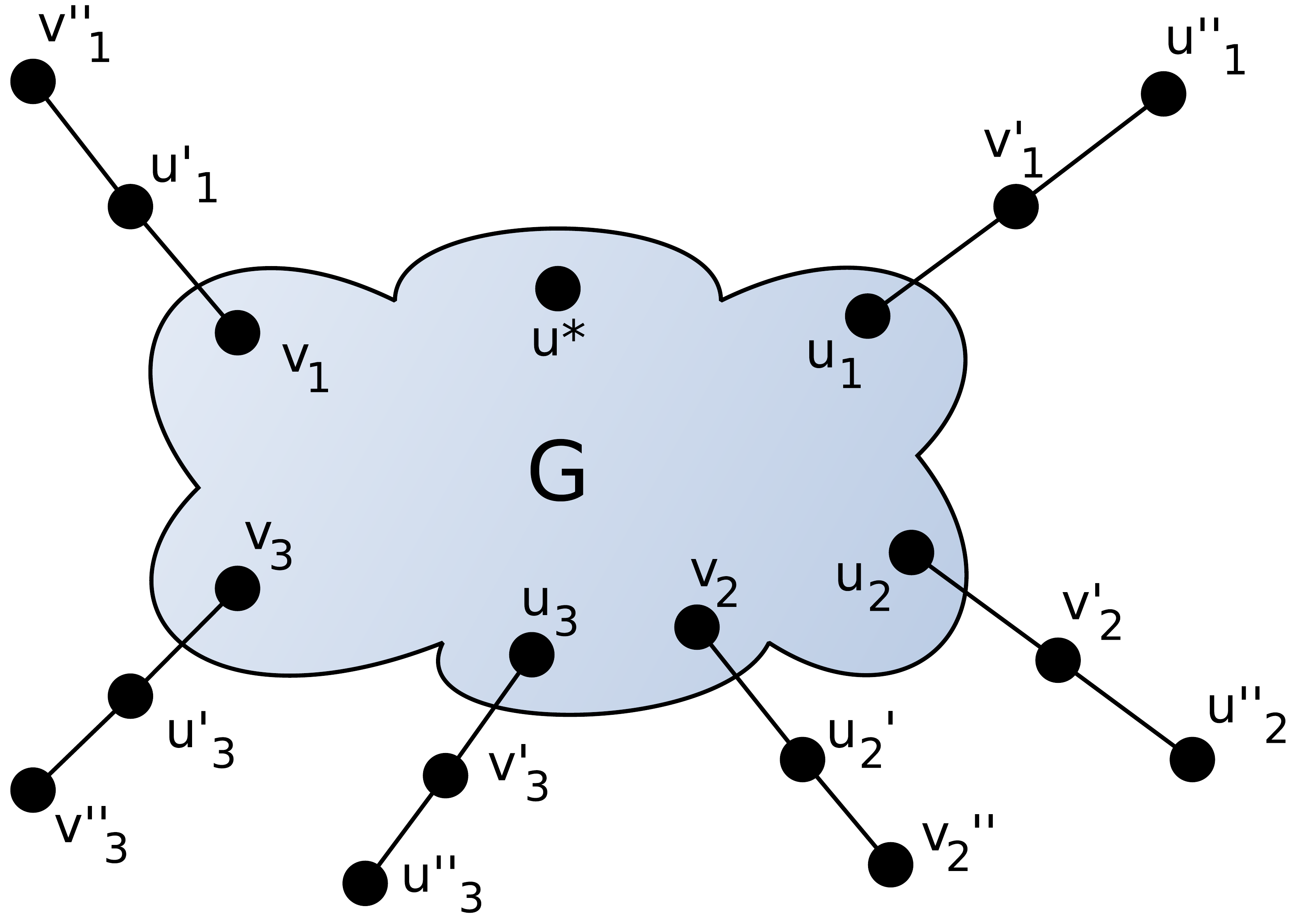}
\caption{Schematic picture for the construction of graph~$\tilde{G}$}
\label{figure:graphSharpP}
\end{figure}

Notice first that~$|M(\tilde{G})| = |M(G)|$. The number of near-perfect matchings of~$\tilde{G}$ is
\begin{eqnarray*}
|N(\tilde{G})| &=& \sum_{u \in \tilde{U}, v \in \tilde{V}} |N_{u,v}(\tilde{G})| \\
&=& \phantom{+} \sum_{u \in U, v \in V} |N_{u,v}(\tilde{G})| + \sum_{u \in U, v' \in V'} |N_{u,v'}(\tilde{G})| + \sum_{u \in U, v'' \in V''} |N_{u,v''}(\tilde{G})| \\
&& + \sum_{u' \in U', v \in V} |N_{u',v}(\tilde{G})| + \sum_{u' \in U', v' \in V'} |N_{u',v'}(\tilde{G})| + \sum_{u' \in U', v'' \in V''} |N_{u',v''}(\tilde{G})| \\
&& + \sum_{u'' \in U'', v \in V} |N_{u'',v}(\tilde{G})| + \sum_{u'' \in U'', v' \in V'} |N_{u'',v'}(\tilde{G})| + \sum_{u'' \in U'', v'' \in V''} |N_{u'',v''}(\tilde{G})|.
\end{eqnarray*}
Note that the 2nd, 4th and 5th summand are zero because an unmatched~$v_i' \in V'$ leads to an unmatched~$u_i'' \in U''$ and an unmatched~$u_i' \in U'$ to an unmatched~$v_i'' \in V''$. We now look closer to the remaining six non-zero summands. 
\begin{enumerate}[label={\alph*)}]
\item The 1st summand easily reduces to~$\sum_{u \in U, v \in V} |N_{u,v}(G)| = |N(G)|$.
\item The 3rd summand also equals~$|N(G)|$ because an unmatched~$v_i'' \in V_i''$ forces~$v_i \in V$ to be matched with~$u_i'$ so no matching in~$G$ can use~$v_i$.
\item Choosing~$u_i' \in U'$ and~$v_j'' \in V''$ in the 6th summand leads to two different cases:
\[ |N_{u_i',v_j''}(\tilde{G})| = 
	\begin{cases}
		0 & i \not= j \\
		|M(\tilde{G})| & i = j
	\end{cases}
\] There are~$n$ ways to choose~$i=j$ so the 6th summand reduces to~$n|M(G)|$.
\item The 7th summand is nearly symmetric to the 3rd, with the difference that vertex~$u^*$ has to be considered separately. So the 7th summand equals~$|N(G)| - \sum_{v\in V}|N_{u^*,v}(G)|$.
\item With the same argumentation as in case c) the 8th summand equals~$(n-1)|M(G)|$.
\item The 9th summand reduces to the 7th one and so to~$|N(G)| - \sum_{v\in V}|N_{u^*,v}(G)|$.
\end{enumerate}
Putting together all the terms we get 
\[ \frac{|N(\tilde{G})|}{|M(\tilde{G})|} =  4 \cdot \frac{|N(G)|}{|M(G)|} + (2n-1) - 2\cdot\sum_{v\in V} \frac{|N_{u^*,v}(G)|}{|M(G)|}. \]

Assume we could efficiently compute the fraction~$|N(G)| / |M(G)|$ for arbitrary bipartite graphs, so in particular for~$G$ and~$\tilde{G}$, we could efficiently compute~$\sum_{v \in V} \frac{|N_{u^*,v}(G)|}{|M(G)|}$ for arbitrary vertex~$u^*$ in contradiction to Proposition \ref{prop2}.\qed
\end{proof}

Figure~\label{fig:experiment1_eps1E9} shows the results of our first experiment for a far smaller $\epsilon$ of $10^{-9}$. The effects as discussed in Section~\ref{sec:experiments} are confirmed.
Figure~\label{fig:experiment2_eps1E9} shows the results of the second experiment for $\epsilon=10^{-9}$. Note, that the spectral bound becomes very tight for such a small $\epsilon$.

\begin{figure}
	\centering
	\begin{subfigure}[b]{0.49\textwidth}
		\includegraphics[width=\textwidth]{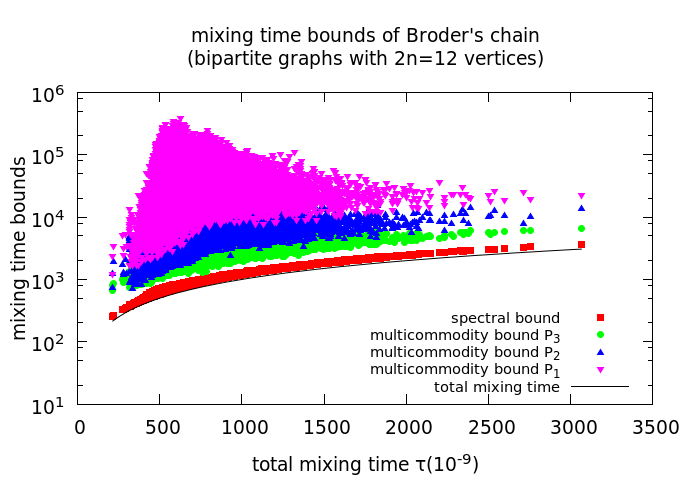}
	\end{subfigure}
	\begin{subfigure}[b]{0.49\textwidth}
		\includegraphics[width=\textwidth]{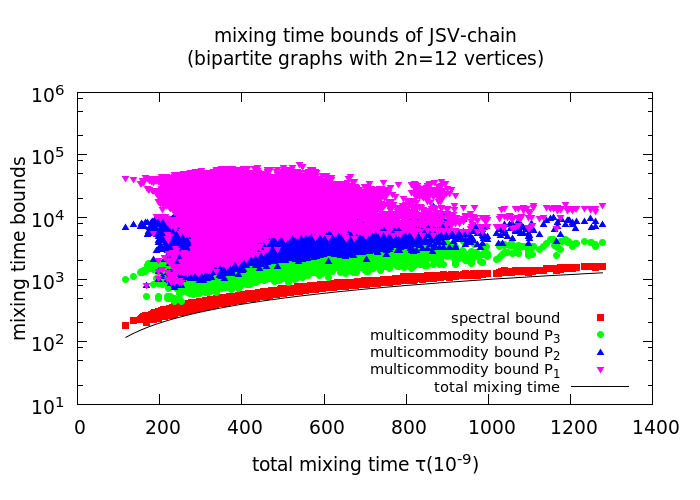}
	\end{subfigure}
	\caption{Results of the first experiment with $\epsilon=10^{-9}$}
	\label{fig:experiment1_eps1E9}
\end{figure}

\begin{figure}
	\centering
	\begin{subfigure}[b]{0.49\textwidth}
		\includegraphics[width=\textwidth]{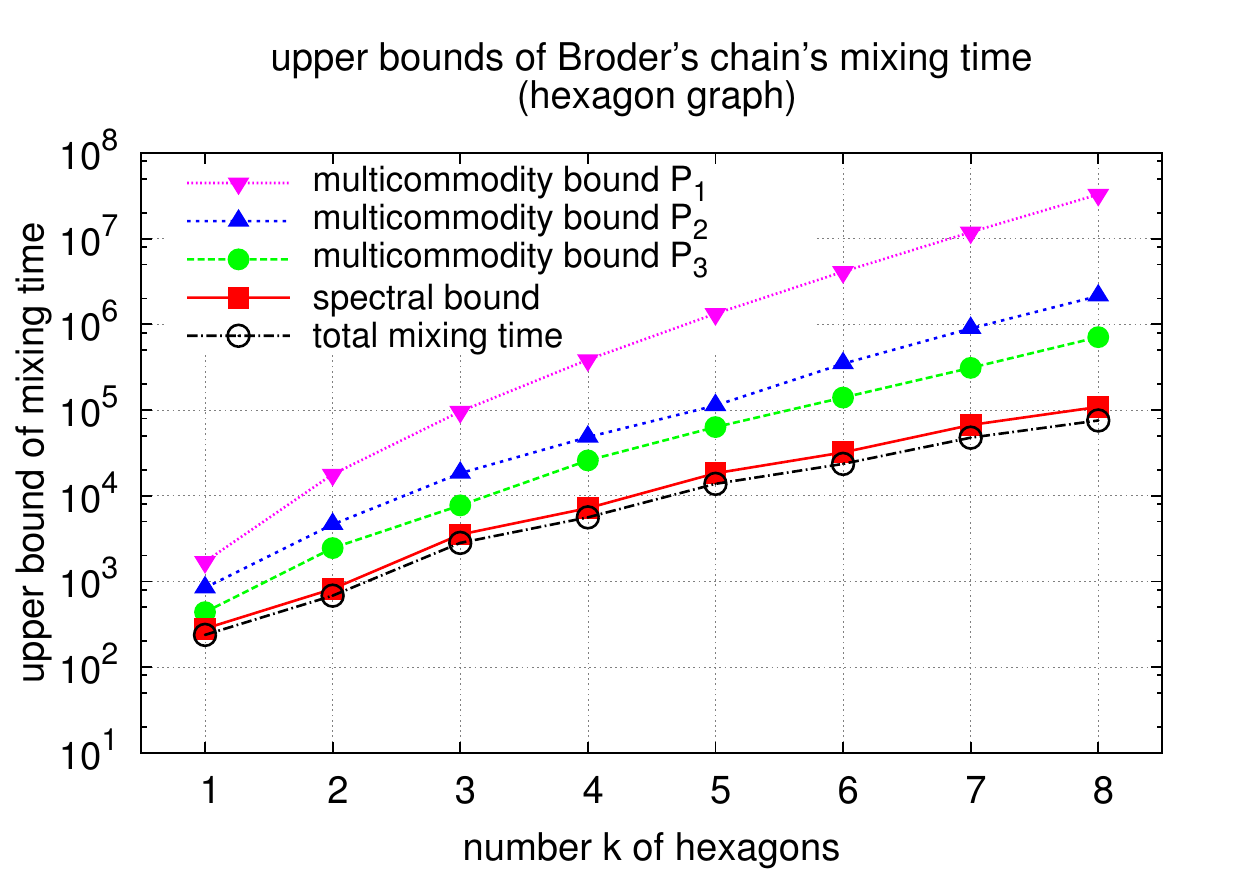}
		\caption{Broder's chain}
	\end{subfigure}
	\begin{subfigure}[b]{0.49\textwidth}
		\includegraphics[width=\textwidth]{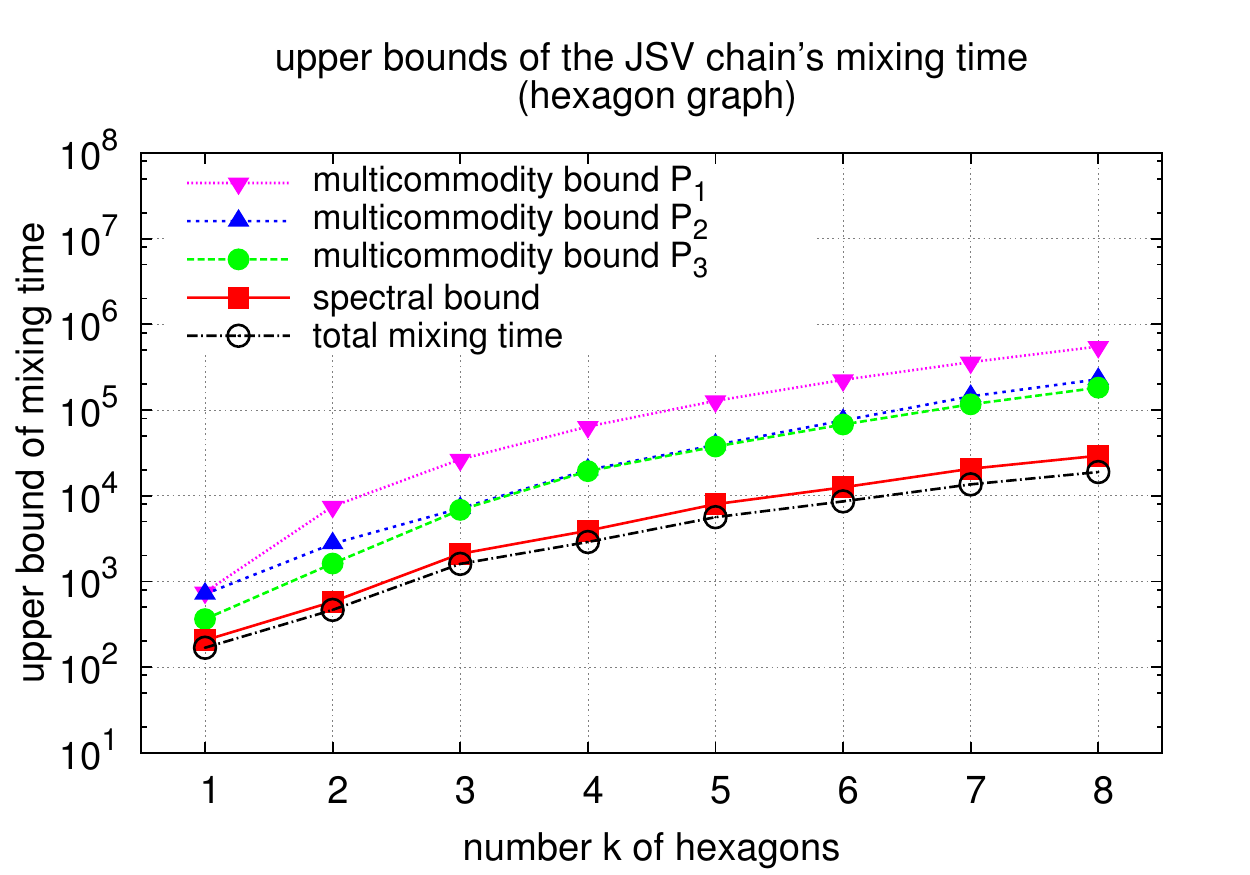}
		\caption{JSV's chain}
	\end{subfigure}
	\caption{Upper bounds of mixing time for hexagon graphs with small~$k$ for $\epsilon=10^{-9}$}
	\label{fig:experiment2_eps1E9}
\end{figure}

\end{appendix}

\end{document}